\pdfoutput=1
\RequirePackage{ifpdf}
\ifpdf 
\documentclass[pdftex]{sigma}
\else
\documentclass{sigma}
\fi

\usepackage[all]{xy}
\usepackage{multirow}

\def\C{{\mathbb C}}

\numberwithin{equation}{section}

\newtheorem{Theorem}{Theorem}[section]
\newtheorem{Corollary}[Theorem]{Corollary}
\newtheorem{Proposition}[Theorem]{Proposition}
{ \theoremstyle{definition}
\newtheorem{Definition}[Theorem]{Definition}
\newtheorem{Note}[Theorem]{Note}
\newtheorem{Remark}[Theorem]{Remark} }

\begin{document}

\allowdisplaybreaks

\newcommand{\arXivNumber}{1708.02280}

\renewcommand{\PaperNumber}{099}

\FirstPageHeading

\ShortArticleName{Contractions of Degenerate Quadratic Algebras, Abstract and Geometric}

\ArticleName{Contractions of Degenerate Quadratic Algebras,\\ Abstract and Geometric}

\Author{Mauricio A.~ESCOBAR RUIZ~$^\dag$, Willard MILLER Jr.~$^\ddag$ and Eyal SUBAG~$^\S$}

\AuthorNameForHeading{M.A.~Escobar~Ruiz, W.~Miller, Jr.\ and E.~Subag}

\Address{$^\dag$~Centre de Recherches Math\'ematiques, Universit\'e de Montreal, \\
\hphantom{$^\dag$}~C.P.~6128, succ. Centre-Ville, Montr\'eal, QC H3C 3J7, Canada}
\EmailD{\href{mailto:mauricio.escobar@nucleares.unam.mx}{mauricio.escobar@nucleares.unam.mx}}

\Address{$^\ddag$~School of Mathematics, University of Minnesota, Minneapolis, Minnesota, 55455, USA}
\EmailD{\href{mailto:miller@ima.umn.edu}{miller@ima.umn.edu}}
\URLaddressD{\url{https://www.ima.umn.edu/~miller/}}
\Address{$^\S$~Department of Mathematics, Pennsylvania State University, State College,\\
\hphantom{$^\S$}~Pennsylvania, 16802 USA}
\EmailD{\href{mailto:eus25@psu.edu}{eus25@psu.edu}}

\ArticleDates{Received August 09, 2017, in f\/inal form December 26, 2017; Published online December 31, 2017}

\Abstract{Quadratic algebras are generalizations of Lie algebras which include the symmetry algebras of 2nd order superintegrable systems in 2 dimensions as special cases. The superintegrable systems are exactly solvable physical systems in classical and quantum mechanics. Distinct superintegrable systems and their quadratic algebras can be related by geometric contractions, induced by B\^ocher contractions of the conformal Lie algebra $\mathfrak{so}(4,\mathbb {C})$ to itself. In 2 dimensions there are two kinds of quadratic algebras, nondegenerate and degenerate. In the geometric case these correspond to 3 parameter and 1 parameter potentials, respectively. In a previous paper we classif\/ied all abstract parameter-free nondegenerate quadratic algebras in terms of canonical forms and determined which of these can be realized as quadratic algebras of 2D nondegenerate superintegrable systems on constant curvature spaces and Darboux spaces, and studied the relationship between B\^ocher contractions of these systems and abstract contractions of the free quadratic algebras. Here we carry out an analogous study of abstract parameter-free degenerate quadratic algebras and their possible geometric realizations. We show that the only free degenerate quadratic algebras that can be constructed in phase space are those that arise from superintegrability. We classify all B\^ocher contractions relating degenerate superintegrable systems and, separately, all abstract contractions relating free degenerate quadratic algebras. We point out the few exceptions where abstract contractions cannot be realized by the geometric B\^ocher contractions.}


\Classification{22E70; 16G99; 37J35; 37K10; 33C45; 17B60; 81R05; 33C45}

\section{Introduction}\label{int1}
An abstract {\it degenerate (quantum) quadratic algebra}~$Q$
is a noncommutative multiparameter associative algebra generated by linearly independent operators $X$, $H$, $L_1$, $L_2$, with parameters~$a_i$,
such that $H$ is in the center and the following commutation relations hold~\cite{KMP2014}:
 \begin{gather}\label{structure2}[X,L_j]=\sum_{0\leq e_1+e_2+e_3+e_4\leq 1} P^{(j)}_{e_1,e_2,e_3, e_4} L_1^{e_1}L_2^{e_2}H^{e_3}X^{2e_4} ,\qquad j=1,2,\\
 \label{commutator1}[L_1,L_2]=\sum_{0\leq e_1+e_2+e_3+e_4\leq 1} T_{e_1,e_2, e_3, e_4}\{ L_1^{e_1},L_2^{e_2},X\}H^{e_3}X^{2e_4}. \end{gather}
Finally, there is the relation:
\begin{gather}\label{Casimir2}G\equiv\sum_{0\leq e_1+e_2+e_3+e_4\leq 2} S_{e_1, e_2, e_3, e_4}\big\{L_1^{e_1},L_2^{e_2}, X^{2e_4}\big\}H^{e_3}+c_1XL_1X+c_2XL_2X=0, \\
 X^0=H^0=I,\nonumber\end{gather}
where $\{L_1^{e_1},L_2^{e_2}, X^{2e_4}\}$ is the 6-term symmetrizer of three operators. The constants
$P^{(j)}_{e_1, e_2, e_3, e_4}$, $T_{e_1, e_2, e_3, e_4}$ and $S_{e_1, e_2, e_3, e_4}$
are polynomials in the parameters $a_i$ of degrees $1-e_1-e_2-e_3-e_4,$ $1-e_1-e_2-e_3-e_4$ and $2-e_1-e_2-e_3-e_4$, respectively, while $c_1$, $c_2$ are of degree~0. If all parameters $a_j=0$ the algebra is {\it free}.
For these quantum quadratic algebras there is a natural grading such that the operators $H$, $L_j$ are 2nd order and~$X$ is 1st order.
The f\/ield of scalars can be either~$\mathbb{ R}$ or~$\mathbb C$.

 An abstract {\it degenerate $($classical$)$ quadratic algebra} ${\mathcal Q}$
is a Poisson algebra with linearly independent generators ${\mathcal X}$, ${\mathcal H}$, ${\mathcal L}_1$, ${\mathcal L}_2$, and parameters $a_i$, satisfying relations (\ref{structure2}), (\ref{commutator1}), (\ref{Casimir2}) with the commutator replaced by the Poisson bracket, $H$, $L_j$, $X$ by ${\mathcal H}$, ${\mathcal L}_j$, ${\mathcal X} $, and the symmetrizer $\{ L_1^{e_1},L_2^{e_2},X^{e_3}\}$ by the product
${\mathcal L}_1^{e_1}{\mathcal L}_2^{e_2}{\mathcal X}^{e_3}/3!$.

These structures arise naturally in the study of classical and quantum superintegrable systems in two dimensions, e.g., \cite{SCQS, MPW2013}, and, in the case of zero potential systems, they are examples of Poisson structures, on which there is a considerable literature~\cite{DufZun05, GraMarPer93,LauPiVan13}.
 A quantum 2D superintegrable system is an integrable Hamiltonian system on a $2$-dimensional real or complex Riemannian manifold
with potential: $ H=\Delta_2+V$, that admits~$3$ algebraically independent partial dif\/ferential operators commuting with $H$, the maximum possible:
 \begin{gather*} [H,L_j]=0,\qquad L_{3}=H, \qquad j=1,2,3.\end{gather*} Here~$\Delta_2$ is the Laplace operator on the manifold. (We call this a {\it Helmholtz superintegrable system} with eigenvalue equation $H\Psi =E\Psi$ to distinguish it from a Laplace conformally superintegrable system, $H\Psi =(\Delta_2 +V)\Psi=0$~\cite{KMS2016}.) A~system is of order $K$ if the maximum order of the symmetry  operators~$L_j$ (other than $H$) is~$K$; all such systems are known for $K=1,2$~\cite{DASK2007, KKM2005a,KKMW,KKMP}. Superintegrability captures the properties of quantum Hamiltonian systems that allow the Schr\"odinger eigenvalue problem $H\Psi=E\Psi$ to be solved exactly, analytically and algebraically. A classical 2D superintegrable system is an integrable Hamiltonian system on a~real or complex $2$-dimensional Riemannian manifold with potential: $ {\mathcal H}=\sum\limits_{j,k=1}^2g^{jk}({\bf x})p_jp_k+V({\bf x})$
in local coordinates $x_1$, $x_2$, $p_1$, $p_2$ that admit~3 functionally independent phase space functions~${\mathcal H}$, ${\mathcal L}_1$, ${\mathcal L}_2$
in involution with~${\mathcal H}$, the maximum possible.
\begin{gather*} \{{\mathcal H},{\mathcal L}_j\}=0,\qquad {\mathcal L}_{3}={\mathcal H}, \qquad j=1,2,3.\end{gather*}
A system is of order $K$ if the maximum order of the constants of the motion~$L_j$, $j\ne 3$, as polynomials in $p_1$, $p_2$ is $K$. Again all such systems are known for $K=1,2$, and, for them, there is a 1-1 relationship between classical and quantum 2nd order 2D superintegrable systems~\cite{KKM20041++}, i.e., the quantum system can be computed from the classical system, and vice versa.

The possible superintegrable systems divide into six classes:
\begin{enumerate}\itemsep=0pt
 \item First order systems. These are the (zero-potential) Laplace--Beltrami eigenvalue equations on constant curvature spaces. The symmetry algebras close under commutation to form the Lie algebras ${\mathfrak e}(2,\mathbb{ R})$, ${\mathfrak e}(1,1)$, ${\mathfrak o}(3,\mathbb{ R})$ or ${\mathfrak o}(2,1)$. Such systems have been studied in detail, using group theory methods.
 \item Free triplets. These are superintegrable systems with zero potential and all generators of 2nd order. The possible spaces for which these systems can occur were classif\/ied by Koenigs~(1896). They are: constant curvature spaces, the four Darboux spaces, and eleven 4-parameter Koenigs spaces~\cite{Koenigs}. In most cases the symmetry operators will not generate a quadratic algebra, i.e., the algebra will not close. If the system generates a~nondegenerate quadratic algebra we call it a {\it free quadratic triplet}.
 \item Nondegenerate systems. These are superintegrable systems with a non-zero potential and the generating symmetries are all of 2nd order. The space of potentials is 4-dimensional:
 \begin{gather*}V({\bf x})= a_1V_{(1)}({\bf x})+a_2V_{(2)}({\bf x})+a_3V_{(3)}({\bf x})+a_4.\end{gather*}
 The symmetry operators generate a nondegenerate quadratic algebra with parameters $a_j$.
 \item Degenerate systems. There are 4 generators: one 1st order $X$ and 3 second order $H$, $L_1$, $L_2$. Here, $X^2$ is not contained in the span of $H$, $L_1$, $L_2$. The space of potentials is 2-dimensional: $V({\bf x})= a_1V_{(1)}({\bf x})+a_2$. The symmetry operators generate a degenerate quadratic algebra with parameters~$a_j$. Relation~(\ref{Casimir2}) is an expression of the fact that 4 symmetry operators cannot be algebraically independent. The possible degenerate systems, classif\/ied up to conjugacy with respect to the symmetry groups of their underlying spaces, are listed in Appendix~\ref{appendixA}.
 \item Exceptional system. $ E15$: $V=f(x-iy)$, $f$ an arbitrary function.

The exceptional case is characterized by the fact that the symmetry generators  are functionally linearly dependent~\cite{KKM20041,KKM20041+,KKM20041++,KKMP}. This is the only 2nd order functionally linearly dependent 2D system but there are many such systems in 3D, including the Calogero 3-body system on the line. In 3D such systems have not yet been classif\/ied.
\end{enumerate}
Every degenerate superintegrable system occurs as a restriction of the 3-parameter potentials to 1-parameter ones, such that one of the symmetries becomes a perfect square: $L=X^2$. Here~$X$ is a f\/irst order symmetry and a new 2nd order symmetry appears so that this restriction admits more symmetries than the original system, see Remark~\ref{remark2}. Basic results that relate these superintegrable systems are the closure theorems:
 \begin{Theorem} A free triplet, classical or quantum, extends to a superintegrable system with potential if and only if it generates a free quadratic algebra $\tilde Q$, degenerate or nondegenerate.
 \end{Theorem}

\begin{Theorem} A superintegrable system, degenerate or nondegenerate, classical or quantum, with quadratic algebra~$Q$, is uniquely determined by its free quadratic algebra~$\tilde Q$.
 \end{Theorem}

These theorems were proved for systems in~\cite{KM2014}.  The proofs are constructive: Given a free quadratic algebra $\tilde Q$ one can compute the potential~$V$ and the symmetries of the quadratic algebra~$Q$. Thus as far as superintegrable systems are concerned, all information about the systems is contained in the free classical quadratic algebras.

\begin{Remark} This paper is a companion to~\cite{EKMS2017} where we studied nondegenerate quadratic algebras, and we assume that the reader has some familiarity with this prior work. In particular, B\^ocher contractions, their properties and associated notation, are treated there and we use them in this paper without detailed comment.\end{Remark}

The layout of this paper is as follows: In Section~\ref{section2} we show how degenerate Helmholtz superintegrable systems can be split into St\"ackel equivalence classes
of Laplace conformally superintegrable systems and we determine how each Helmholtz system can be characterized in its equivalence class. In Section~\ref{section3} we determine
all B\^ocher, i.e., geometrical, contractions of the Laplace systems and obtain complete lists of the possible Helmholtz contractions. In Section~\ref{section4} we classify all abstract free quadratic algebras and determine which of these can be realized as the quadratic algebra of a Helmholtz degenerate superintegrable system. In Section~\ref{section5} we classify all abstract contractions of abstract free quadratic algebras and determine which of these can be realized as B\^ocher and Heisenberg contractions of the quadratic algebras of Helmholtz degenerate superintegrable systems. In Fig.~\ref{caption1} we describe how restrictions of nondegenerate superintegrable systems to degenerate ones, and contractions of degenerate superintegrable systems account for the lower half of the Askey scheme. The upper half of the Askey scheme is described by contractions of nondegenerate systems~\cite{EKMS2017}.
In Section~\ref{section6} we assess our results. A list of all Helmholtz degenerate superintegrable systems can be found in Appendix~\ref{appendixA}.

\section{St\"ackel transforms and Laplace equations}\label{section2}
\looseness=-1 Distinct degenerate classical or quantum superintegrable systems can be mapped to one another by St\"ackel transforms, invertible transforms that preserve the structure of the quadratic algebra. This divides the 15 systems into 6 St\"ackel equivalence classes~\cite{MPW2013}. The most convenient way to understand the equivalence classes is in terms of Laplace-like equations~\cite{KMS2016}. Since every 2D space is conformally f\/lat there always exist ``Cartesian-like'' coordinates $x$, $y$ such that the Hamilton--Jacobi equation can be expressed in the form $H=E$ where $H=\frac{p_x^2+p_y^2}{\lambda(x,y)}+\alpha V$ and $\alpha$ is a parameter. This is equivalent to the Laplace-like equation $p_x^2+p_y^2 +a_1 V_1+a_2 V_2=0$ where $V_1=\lambda V$, $V_2=\lambda$, $a_1=\alpha$, $a_2=-E$, now with 2 parameters. Symmetries (constants of the motion) for the Helmholtz equation correspond to conformal symmetries of the Laplace equation. The Hamilton--Jacobi equation is def\/ined on one of a variety of conformally f\/lat spaces but the Laplace equation is always def\/ined on f\/lat space with conformal symmetry algebra ${\mathfrak{so}}(4,\mathbb{ C})$~\cite{KMS2016}. An important observation is that the Laplace equations are St\"ackel equivalence classes: two Helmholtz systems are St\"ackel equivalent if and only if they correspond to the same Laplace equation.

\begin{Remark} Indeed, If the Laplace conformally superintegrable equation can be split in the form $p_x^2+p_y^2 +V_0-{\tilde E}W=0$, where $\tilde E$ is an arbitrary parameter, $W$ is a nonconstant function, and~$V_0$,~$W$ are independent of~$E$, then~$W$, by division, def\/ines a conformal St\"ackel transform to the superintegrable Helmholtz system ${\tilde H}=\frac{1}{W}(p_x^2+p_y^2 +V_0)={\tilde E}$. If the Laplace system admits another splitting $p_x^2+p_y^2 +V'_0-{\tilde E'}W'=0$, it determines another superintegrable Helmholtz system ${\tilde H'}=\frac{1}{W'}(p_x^2+p_y^2 +V'_0)={\tilde E'}$ and $\tilde H'$ can be obtained from ${\tilde H}$ by an invertible St\"ackel transform $\frac{W'}{W}$. Thus all Helmholtz systems that can be obtained from the Laplace equation by splitting the potential are St\"ackel equivalent to one another. \end{Remark}

The Laplace equations for nondegenerate systems were derived in~\cite{KMS2016}, see Table~\ref{Tab1}. The Laplace equations for degenerate systems are listed in Table~\ref{degpot}. The notation~$a_i$ in Table~\ref{degpot} describes how these systems can be obtained as restrictions of systems in Table~\ref{Tab1}, but with added symmetry.

\begin{table}[t!]\centering
\begin{tabular}{|c||c| }
\hline
System & Non-degenerate potentials $V(x, y)$\bsep{2pt}\tsep{2pt} \\
\hline
\hline
$ [1,1,1,1] $ & $
\frac{a_1}{x^2}+\frac{a_2}{y^2}+\frac{4 a_3}{(x^2+y^2-1)^2}-\frac{4 a_4}{(x^2+y^2+1)^2}$\tsep{3pt}\bsep{3pt}\\
$ [2,1,1] $ & $ \frac{a_1}{x^2} + \frac{a_2}{y^2} - a_3 (x^2+y^2) + a_4 $\bsep{3pt}\\
$ [2,2] $ & $
\frac{a_1}{(x+i y)^2}+\frac{a_2 (x-i y)}{(x+i y)^3}+a_3-a_4 (x^2+y^2) $\bsep{3pt}\\
$ [3,1] $ & $ a_1-a_2 x+a_3 (4 x^2+y^2)+\frac{a_4}{y^2} $\bsep{3pt}\\
$ [4] $ & $ a_1-a_2 (x+i y)
+a_3 (3(x+i y)^2+2(x-i y))-a_4 (4(x^2+y^2)+2(x+i y)^3) $\bsep{3pt}\\
$ [0] $ & $ a_1-(a_2 x+a_3 y)+a_4 (x^2+y^2) $\bsep{3pt}\\
$ (1) $ & $ \frac{a_1}{(x+i y)^2}+a_2
-\frac{a_3}{(x+i y)^3}+\frac{a_4}{(x+i y)^4} $\bsep{3pt}\\
$ (2) $ & $ a_1+a_2(x+i y)+a_3(x+i y)^2+a_4(x+i y)^3 $\bsep{3pt}\\
\hline
\end{tabular}
\caption{Four parameter Laplace systems.}\label{Tab1}
\end{table}

\begin{table}[h!]\centering
\begin{tabular}{|c||c|}
\hline
System & Degenerate potentials  $V(x,y)$\tsep{2pt}\bsep{2pt}\\
\hline
\hline
$ A $ & $ \frac{4\,a_3}{(x^2+y^2-1)^2}-\frac{4a_4}{(x^2+y^2+1)^2} $\tsep{2pt}\bsep{2pt}\\
$ B $ & $ \frac{a_1}{x^2}+a_4 $\bsep{2pt}\\
$ C $ & $ a_3-a_4(x^2+y^2) $\bsep{2pt}\\
$ D $ & $ a_1-a_2x $\bsep{2pt}\\
$ E $ & $ \frac{a_1}{(x+iy)^2}+a_3 $\bsep{2pt}\\
$ F $ & $ a_1-a_2(x+iy) $\bsep{2pt}\\
\hline
\end{tabular}
\caption{Two-parameter Laplace systems.}\label{degpot}
\end{table}

The Helmholtz systems corresponding to each Laplace system are:

{\bf St\"ackel equivalence classes:} Here the notation refers to the Helmholtz superintegrable systems listed in the Appendix.
\begin{enumerate}\itemsep=0pt
 \item Class $A$ $(a_3,a_4)$: System $S3$ corresponds to $(1,0)$ and $(0,1)$. System $S6$ corresponds to $(1,1)$.
 System $D4(b)D$ corresponds to $(a_3,a_4)$ with $a_3a_4(a_3-a_4)\ne 0$.
 \item Class $B$\ $(a_1,a_4)$: System $S5$ corresponds to $(1,0)$. System $E6$ corresponds to $(0,1)$. System $D2D$ corresponds to $(a_1,a_4)$
 with $a_1a_4\ne 0$.
 \item Class $C$\ $(a_3,a_4)$: System $E3$ corresponds to $(1,0)$. System $E{18}$ corresponds to $(0,1)$. System $D3E$ corresponds to $(a_3,a_4)$
 with $a_3a_4\ne 0$.
 \item Class $D$\ $(a_1,a_2)$: System $E5$ corresponds to $(1,0)$. System $D1D$ corresponds to $(a_1,a_2)$
 with $a_1a_2\ne 0$.
 \item Class $E$\ $(a_1,a_3)$: System $E{14}$ corresponds to $(0,1)$. System $E{12}$ corresponds to $(a_1,a_3)$
 with $a_1a_3\ne 0$.
 \item Class $F$\ $(a_1,a_2)$: System $E{13}$ corresponds to $(a_1,a_2)$
 with $a_2\ne 0$. System $E{4}$ corresponds to $(1,0)$.
 \end{enumerate}
Here, for example, system $D3E$ belongs to class $C$ and is obtained from the Laplace equation by dividing it by
$a_3-a_4(x^2+y^2)$ where $a_3a_4\ne 0$, whereas $E{18}$ is obtained by the same division with $a_3=0$, $a_4=1$.

The conformal symmetry of these Laplace equations is best exploited by using tetraspherical coordinates to linearize the
action of the conformal symmetry group~\cite{Bocher,KMS2016}. These are projective coordinates $x_1$, $x_2$, $x_3$, $x_4$ on the null
cone $x_1^2+x_2^2+x_3^2+x_4^2=0$, related to f\/lat space coordina\-tes~$x$,~$y$ by
\begin{gather*} x_1=2XT,\qquad x_2=2YT,\qquad x_3=X^2+Y^2-T^2,\qquad x_4=i\big(X^2+Y^2+T^2\big),\\
x=\frac{X}{T}=-\frac{x_1}{x_3+ix_4},\qquad y=\frac{Y}{T}=-\frac{x_2}{x_3+ix_4},\qquad
 x=\frac{s_1}{1+s_3},\qquad y=\frac{s_2}{1+s_3},\\
  {\mathcal H}\equiv p_x^2+p_{y}^2+{ V}=(x_3+ix_4)^2\left(\sum_{k=1}^4p_{x_k}^2+V_B\right)
=(1+s_3)^2\left(\sum_{j=1}^3p_{s_j}^2+V_S\right),\\
V=(x_3+ix_4)^2V_B, \!\!\!\qquad
 (1+s_3)=-i\frac{(x_3+ix_4)}{x_4},\!\!\!\qquad
s_1=\frac{ix_1}{x_4},\!\!\!\qquad s_2=\frac{ix_2}{x_4},\!\!\!\qquad s_3=\frac{-ix_3}{x_4}.\end{gather*}
Thus the Laplace equation $ {\mathcal H}\equiv p_x^2+p_{y}^2+{V}=0$ in Cartesian coordinates
becomes $\sum\limits_{k=1}^4p_{x_k}^2+V_B=0$ in tetraspherical coordinates. Here, the $s_j$ refer to coordinates on the unit 2-sphere: \smash{$s_1^2+s_2^2+s_3^2=1$},

The possible limits of one superintegrable system to another can be derived and classif\/ied by using tetraspherical coordinates and special B\^ocher contractions of ${\mathfrak{so}}(4,\mathbb{ C})$ to itself. The method is described in detail in~\cite{EKMS2017, KMS2016}. Here we just recall the basic def\/inition of a B\^ocher contraction, Let ${\bf x}={\bf A}(\epsilon){\bf y}$, and ${\bf x}=(x_1,\dots,x_4)$, ${\bf y}=(y_1,\dots,y_4)$ be column vectors, and ${\bf A}=(A_{jk}(\epsilon))$ be a~$4\times 4$ matrix with matrix elements
\begin{gather*}
A_{kj}(\epsilon)=\sum_{\ell=-N}^Na^\ell_{kj}\epsilon^\ell,\end{gather*} where $N$ is a nonnegative integer and the $a^\ell_{kj}$ are complex constants.
(Here, $N$ can be arbitrarily large, but it must be f\/inite in any particular case.) We say that the matrix $\bf A$ def\/ines a {\it B\^ocher contraction} of the conformal algebra ${\mathfrak{so}}(4,\mathbb {C})$ to itself provided
\begin{gather*}
1)\quad
\det ({\bf A})=\pm 1, \ {\rm constant\ for\ all\ }\epsilon\ne 0,\\
 2)\quad{\bf x}\cdot{\bf x}\equiv \sum_{j=1}^4x_i(\epsilon)^2={\bf y}\cdot{\bf y}+O(\epsilon).\end{gather*}
If, in addition, ${\bf A}\in O(4,\mathbb {C})$ for all $\epsilon\ne 0$ the matrix $\bf A$ def\/ines a {\it special B\"ocher contraction}. For a special B\"ocher
contraction ${\bf x}\cdot{\bf x}={\bf y}\cdot{\bf y}$, with no error term.
 (These contractions correspond to limit relations introduced by B\^ocher to obtain all orthogonal separable coordinates for Laplace and wave equations as limits of cyclidic coordinates. There is an inf\/inite family of such contractions, but they can be generated by 4 basic contractions.) Related contraction methods that don't make use of tetraspherical coordinates directly can be found in references such as~\cite{Pog01, Pog96}.

B\^ocher contractions take a Laplace system to itself. The contraction process has already been described in~\cite{EKMS2017, KMS2016} and references therein, but we discuss, brief\/ly, the main ideas. Suppose we have a degenerate Laplace superintegrable system with potential $V({\bf x},{\bf a})=a_1 V_1({\bf x})+a_2V_2({\bf x})$ and generating conformal symmetries $X={\cal X}+W_0$, $L_1={\cal L}_1+W_1$, $L_2={\cal L}_2+W_2$, where ${\cal X}$,~${\cal L}_1$,~${\cal L}_2$ are free 2nd order conformal symmetries and $W_0$, $W_1$, $W_2$ are functions of the tetraspherical coordinates~$x_i$. Applying a~B\^ocher contraction $A(\epsilon)$ to the free symmetries we obtain
\begin{gather*}{\cal X}(\epsilon)=\epsilon^{\alpha_0}{\cal X}'+O(\epsilon),\qquad {\cal L}_j(\epsilon)=\epsilon^{\alpha_j}{\cal L}_j'+O(\epsilon),\qquad j=1,2,\end{gather*}
where ${\cal X}'\in {\mathfrak{so}}(4,\C)$ and the ${\cal L}_j'$ are quadratic in ${\mathfrak{so}}(4,\C)$. By a change of basis $\{ L_1,L_2\}$ if necessary, one can verify that ${\cal X}'$, ${\cal L}_j'$, $j=1,2$ generate a free conformal quadratic algebra, The action of the B\^ocher contraction on the 2-dimensional potential space preserves its dimension and maps it smoothly as a function of~$\epsilon$, as follows from an examination of the Bertrand-Darboux equations. Thus we get a 2-dimensional potential space in the limit. To f\/ind an explicit basis for the contracted potential $V'({\bf b},{\bf y})=b_1 V_1'({\bf y})'+b_2V_2({\bf y})'$ we put $a_1=\sum\limits_{k=-\infty}^\infty c_k\epsilon^k$, $a_2=\sum\limits_{k=-\infty}^\infty d_k\epsilon^k$, where only a f\/inite number of the coef\/f\/icients $c_k$, $d_k$ can be nonzero. Then it is a linear algebra problem to determine the $c_k$, $d_k$ such that $\lim\limits_{\epsilon\to 0} V({\bf x}(\epsilon),{\bf a}(\epsilon))= V'({\bf b},{\bf y})$ exists for independent potential functions $V_1'({\bf y})'$, $V_2'({\bf y})'$, where the nonzero $c_k$, $d_k$ are linear in $b_1$, $b_2$. The limit is guaranteed to exist and is unique up to a change of basis $\{V_1'({\bf y})', V_2'({\bf y})'\}$ for the target potential. Only the last limit and the linear algebra problem need to be solved to identify the contraction. This work was carried out with the assistance of the symbol manipulation programs Maple and Mathematica. There is one additional complication; the results of the contraction are not invariant under a permutation of the indices of the hyperspherical coordinates def\/ining the contraction matrix~$A_{ij}$. Thus one B\^ocher contraction applied to a source system can yield a multiplicity of target results, and all permutations need to be examined.

The results are rather complicated. Fig.~\ref{figurea} provides a clearer idea of what is happening.  There are 4 basic B\"ocher contractions of 2D Laplace systems and each one when applied to a~Laplace system, and each permutation treated, yields another Laplace superintegrable system. A system in class $K_1$ can be obtained from a system in class~$K_2$ via contraction provided there is a directed arrow path from~$K_2$ to~$K_1$. All systems follow from $A$ for degenerate potentials, and~$A$ is a~restriction of $[1111]$ with increased symmetry. Fig.~\ref{figurea} describes only the existence or nonexistence of contractions, not the multiplicity of distinct contractions.

\begin{figure}[h]\centering
\includegraphics[width=45mm]{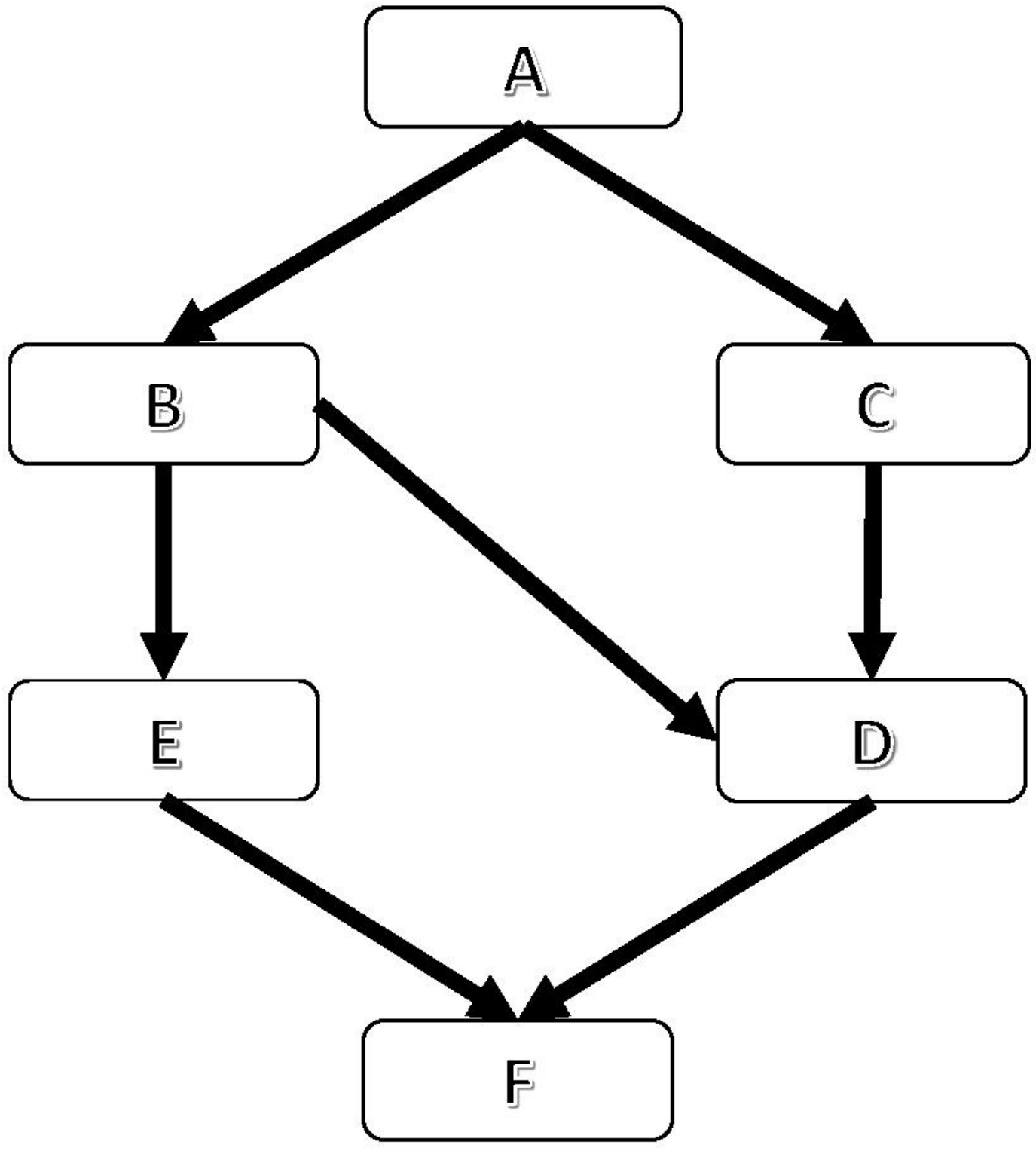}
\caption{Contraction relations for degenerate Laplace systems.}\label{figurea}
\end{figure}

Our basic interest, however, is in Helmholtz contractions, i.e., contractions of a~Helmholtz superintegrable system to another such system. The key is to start with a Laplace system, take a conformal St\"ackel transform to a Helmholtz system (which we initially interpret as another Laplace system) and then take a B\^ocher contraction of the new system, which as described below gives a new Helmholtz system. The result is the contraction of one Helmholtz system to another This can be done in such a~way the ``diagrams commute'', i.e., a Helmholtz contraction is induced by a B\^ocher contraction and a St\"ackel transform~\cite{KMS2016}. For example, let $\cal H$ be the Hamiltonian for class~$A$. In terms of tetraspherical coordinates a general conformal St\"ackel transformed potential will take the form
\begin{gather*} V=\frac{\frac{a_3}{x_3^2}+\frac{a_4}{x_4^2}}{\frac{b_3}{x_3^2}+\frac{b_4}{x_4^2}} =\frac{V_{A}}{F({\bf x},{\bf b})},\end{gather*}
where $ F({\bf x},{\bf b})=\frac{b_3}{x_3^2}+\frac{b_4}{x_4^2}$, and the transformed Hamiltonian will be ${\hat {\cal H}}=\frac{1}{ F({\bf x},{\bf b})}{\cal H}$, where the transform is determined by the f\/ixed vector $(b_3,b_4)$. Now we apply the B\^ocher contraction $[1,1,1,1]\to [2,1,1]$ to this system. Depending on the permutation of the indices~$x_j$, in the limit as $\epsilon\to 0$ the potential $V_{A}\to V_{B}$, or $V_{A}\to V_{C}$, and $\cal H\to {\cal H}'$, the $B$ or $C$ system. Now consider
$ F({\bf x}(\epsilon),{\bf b})= V'({\bf x}',b)\epsilon^\alpha+O\big(\epsilon^{\alpha+1}\big)$,
where the integer exponent $\alpha$ depends upon our choice of $\bf b$. From our theory, the system def\/ined by Hamiltonian ${\hat {\cal H}}'=\lim\limits_{\epsilon\to 0}\epsilon^\alpha {\hat{\cal H}}(\epsilon)=\frac{1}{V'({\bf x}',{\bf b})}{\cal H}'$ is a superintegrable system that arises from the system $A$ by a conformal St\"ackel transform induced by the potential $V'({\bf x}',{\bf b})$. Thus the Helmholtz superintegrable system with potential $V=V_{A}/F$ contracts to the Helmholtz superintegrable system with potential $V_{S}/V'$, where $S=B$ or $S=C$. The contraction is induced by a generalized In\"on\"u--Wigner Lie algebra contraction of the conformal algebra ${\mathfrak{so}}(4,\mathbb{ C})$. Always the $V'$ can be identif\/ied with a specialization of the $S$ potential. Thus a conformal St\"ackel transform of $A$ has been contracted to a conformal St\"ackel transform of~$S$.

\section{Degenerate Helmholtz contractions} \label{Helmholtzcontractions}\label{section3} The superscript for each targeted Helmholtz system is the value of the exponent $\alpha$ associated with the contraction. In each table below, corresponding to a single Laplace equation equivalence class, the top line is
a list of the Helmholtz systems in the class, and the lower lines are the target systems under the B\^ocher contraction.
\begin{gather*}
\begin{array}{llllllll}& A &{\rm equivalence}&{\rm class} & {\rm contractions}& &\\
\hline
{\rm contraction} &{S3}&S6&D4D\\
\hline
{[1111]}\downarrow[211]&E3^2&E{18}^4&E3^2\tsep{2pt}\\
&S5^0, \ E6^2&S5^0&S5^0\\
\hline
{[1111]}\downarrow[22]&E3^2 &E{18}^4 &E{3}^2\tsep{2pt}\\
&E{14}^2&E{12}^2&E{12}^2\\
\hline
{[1111]}\downarrow[31]&E5^2&E5^2&E5^2, \ D1D^3\tsep{2pt}\\
&E6^2, \ S5^0&S5^0&S5^0\\
\hline
{[1111]}\downarrow[4]&E{4}^6&E{4}^6&E{4}^6, \ E{13}^8\tsep{2pt}\\
\hline
\end{array}
\\
\begin{array}{llllllll}& B &{\rm equivalence}&{\rm class} & {\rm contractions}& &\\
\hline
{\rm contraction} &{S5}&E6&D2D\\
\hline
{[1111]}\downarrow[211]&E{14}^2&E{14}^0&E{14}^0\tsep{2pt}\\
&S5^0& E6^2&S5^0\\
\hline
{[1111]}\downarrow[22]&E{14}^2 &E{14}^2 &E{12}^2\tsep{2pt}\\
\hline
{[1111]}\downarrow[31]&E5^2&E5^2&E5^2,\ D1D^3\tsep{2pt}\\
&S5^0&E6^2&S5^0, \ E6^2\\
\hline
{[1111]}\downarrow[4]&E{4}^6&E{4}^6&E{4}^6, \ E{13}^8\tsep{2pt}\\
\hline
\end{array}
\\
\begin{array}{llllllll}& C &{\rm equivalence}&{\rm class} & {\rm contractions}& &\\
\hline
{\rm contraction} &{E3}&E{18}&D3E\tsep{2pt}\\
\hline
{[1111]}\downarrow[211]&E{5}^2&E{5}^2&E{5}^{2,3}\tsep{2pt}\\
&E3^2& E{18}^4&E3^2\\
\hline
{[1111]}\downarrow[22]&E3^2 &E{18}^{2,4} &E3^2, \ D3E^2\tsep{2pt}\\
\hline
[1111]\downarrow[31]&E5^2&E5^2&E5^2, \ D1D^3\tsep{2pt}\\
\hline
{[1111]}\downarrow[4]&E{4}^6&E{4}^6&E{4}^6, \ E{13}^8\tsep{2pt}\\
\hline
\end{array}
\\
\begin{array}{llllllll}& D&{\rm equivalence}&{\rm class}& {\rm contractions}& &\\
\hline
{\rm contraction} &{E5}&D1D\\
\hline
{[1111]}\downarrow[211]&E{4}^0&E{13}^{-1},\ E4^0\tsep{2pt}\\
&E5^2&E5^2,\  D1D^3\\
\hline
{[1111]}\downarrow[22]&E4^2 &E{4}^2,\ E{13}^2\tsep{2pt}\\
\hline
{[1111]}\downarrow[31]&E5^2&E5^2, \ D1D^2\tsep{2pt}\\
\hline
[1111]\downarrow[4]&E{4}^6&E{4}^6, \ E{13}^8\tsep{2pt}\\
\hline
\end{array}
\\
\begin{array}{llllllll}& E&{\rm equivalence}&{\rm class}& {\rm contractions}& &\\
\hline
{\rm contraction} &E{14}&E{12}\\
\hline
{[1111]}\downarrow[211]&E{4}^{2}&E{4}^{2}, \ E{13}^3\tsep{2pt}\\
&E{14}^0&E{14}^0\\
\hline
{[1111]}\downarrow[22]&E4^2 &E{4}^2, \ E{13}^4\tsep{2pt}\\
&E{14}^2&E{12}^2\\
\hline
{[1111]}\downarrow[31]&E4^2&E4^2, \ E{13}^3\tsep{2pt}\\
\hline
[1111]\downarrow[4]&E{4}^6&E{4}^6, \ E{13}^8\tsep{2pt}\\
\hline
\end{array}
\\
\begin{array}{llllllll}& F &{\rm equivalence}&{\rm class} & {\rm contractions}& &\\
\hline
{\rm contraction} &{E4}&E{13}\\
\hline
{[1111]}\downarrow[211]&E{4}^{0,2}&E{13}^{-1,3},\ E4^2\tsep{2pt}\\
\hline
{[1111]}\downarrow[22]&E4^2 &E{13}^2 \tsep{2pt}\\
\hline
{[1111]}\downarrow[31]&E4^2&E4^2\tsep{2pt}\\
\hline
{[1111]}\downarrow[4]&E{4}^6&E{4}^6\tsep{2pt}\\
\hline
\end{array}
\end{gather*}

\section[Classif\/ication of free abstract degenerate classical quadratic algebras]{Classif\/ication of free abstract degenerate\\ classical quadratic algebras}\label{section4}
\subsection{Abstract quadratic algebras}
The special case of a 2D degenerate classical superintegrable system on a constant curvature space or a Darboux space with all parameters equal to zero (no potential) gives rise to a special kind of quadratic algebra which we call a free abstract quadratic algebra. Below is a precise def\/inition.
\begin{Definition}
An abstract 2D free degenerate (classical) quadratic algebra is a complex Poisson algebra possessing a linearly independent generating set $\{\mathcal{L}_1, \mathcal{L}_2, \mathcal{L}_3=\mathcal{H},\mathcal{X}\}$ that satisfy:
\begin{enumerate}\itemsep=0pt
\item The associative product is abelian.
\item $\mathcal{A}$ is graded: $\mathcal{A}=\oplus_{k=0}^{\infty}\mathcal{A}_k$ where each $\mathcal{A}_k$ is a complex vector space, the associative product takes $\mathcal{A}_k \times \mathcal{A}_l$ into $\mathcal{A}_{k+l}$, and the Poisson brackets $\{\_\_,\_\_\}$ goes from $\mathcal{A}_k \times \mathcal{A}_l$ into $\mathcal{A}_{k+l-1}$.
\item $ \mathcal{A}_0=\mathbb{C}$.
\item $\mathcal{A}_1=\mathbb{C}\mathcal{X}$.
\item $\mathcal{A}_2=\operatorname{span}\{\mathcal{L}_1,\mathcal{L}_2,\mathcal{L}_3,\mathcal{X}^2\} $.
\item The elements $\{\mathcal{L}_1, \mathcal{L}_2, \mathcal{L}_3,\mathcal{L}_4:=\mathcal{X}^2\}$ satisfy a relation given by a~ homogeneous polynomial of degree~2, ${G}(\mathcal{L}_1, \mathcal{L}_2, \mathcal{L}_3,\mathcal{L}_4 )=0$.
\item Any nonzero polynomial $F$ of minimal degree such that ${F}(\mathcal{L}_1, \mathcal{L}_2, \mathcal{L}_3,\mathcal{L}_4 )=0$ is a multiple of $G$.
\item The center of $\mathcal{A}$ contains no elements of order $1$ and any element of order two in the center is a multiple of $\mathcal{L}_3=\mathcal{H}$.
\item The polynomial $G$ depends non-trivially on at least one of the non-central 2nd order generators ${\mathcal L}_1$, ${\mathcal L}_2$.
\end{enumerate}
 We shall simply refer to such an algebra as a \textit{free abstract quadratic algebra}.\label{defn1}
\end{Definition}

\begin{Remark}
A free abstract quadratic algebra is a special case of a three-dimensional af\/f\/ine Poisson variety, see, e.g., \cite{DufZun05, LauPiVan13}. Description of some of the properties of such Poisson varieties, along with classif\/ication of such Poisson structures on the af\/f\/ine space $\mathbb{C}^3$ can be found in \cite[Section~9.2]{LauPiVan13}. Poisson structures on the manifold $\mathbb{R}^3$ were classif\/ied at \cite{GraMarPer93}. It should be noted that the Poisson algebras considered in this paper, are not quadratic Poisson structures in the sense of \cite[Section~8.2]{LauPiVan13}. In addition our classif\/ication scheme is dif\/ferent from the above mentioned classif\/ications.
\end{Remark}

\begin{Theorem}
Keeping the same notations as of Definition~{\rm \ref{defn1}}, let $\mathcal{A}$ be a free abstract quadratic algebra. Then there exists a non zero $K\in \mathbb{C}$ such that:{\samepage
\begin{gather}\label{eq}
 \{\mathcal{L}_1, \mathcal{L}_2\}= K\frac{\partial G}{\partial {\cal X}},\qquad \{{\cal X}, \mathcal{L}_1\}= K\frac{\partial G}{\partial\mathcal{L}_2 },\qquad
 \{{\cal X}, \mathcal{L}_2\}= -K\frac{\partial G}{\partial \mathcal{L}_1}.
\end{gather}
$($We shall refer to equations \eqref{eq} as the structure equations of $\mathcal{A})$.}
\end{Theorem}

\begin{proof}  Note that $ {G}(\mathcal{L}_1, \mathcal{L}_2, \mathcal{L}_3,\mathcal{L}_4 )=0$ implies that
\begin{gather*}
0=\{\mathcal{X},G\}=\frac{\partial G}{\partial \mathcal{L}_1}\{{\cal X}, \mathcal{L}_1\}+\frac{\partial G}{\partial \mathcal{L}_2}\{{\cal X}, \mathcal{L}_2\}.
\end{gather*}
The above equation is a polynomial expression in $\chi^2,\mathcal{L}_1$, $\mathcal{L}_2$, $\mathcal{H}$. Any one of the terms $\frac{\partial G}{\partial \mathcal{L}_1}$, $\{{\cal X}, \mathcal{L}_1\}$, $\frac{\partial G}{\partial \mathcal{L}_2}$, $\{{\cal X}, \mathcal{L}_2\}$ is either zero or a polynomial of degree one in the variables $\chi^2$, $\mathcal{L}_1$, $\mathcal{L}_2$,~$\mathcal{H}$. Since at least one of the terms $\frac{\partial G}{\partial \mathcal{L}_1}$, $\frac{\partial G}{\partial \mathcal{L}_2}$ is non-zero then we must have
\begin{gather*}\nonumber
 \{{\cal X}, \mathcal{L}_1\}= K\frac{\partial G}{\partial\mathcal{L}_2 },\qquad
 \{{\cal X}, \mathcal{L}_2\}= -K\frac{\partial G}{\partial \mathcal{L}_1},
\end{gather*}
for some $K$. If $K= 0$ then the f\/irst order element ${\cal X}$ must be in the center of $\mathcal{A}$ which is impossible. Hence $K\neq 0$. Similarly $ {G}(\mathcal{L}_1, \mathcal{L}_2, \mathcal{L}_3,\mathcal{L}_4 )=0$ implies
\begin{gather*}
 0=\{ \mathcal{L}_1,G\}=\frac{\partial G}{\partial \mathcal{X}}\{ \mathcal{L}_1,{\cal X}\}+\frac{\partial G}{\partial \mathcal{L}_2}\{\mathcal{L}_1, \mathcal{L}_2\},\\
 0=\{ \mathcal{L}_2,G\}=\frac{\partial G}{\partial \mathcal{X}}\{ \mathcal{L}_2,{\cal X}\}+\frac{\partial G}{\partial \mathcal{L}_1}\{\mathcal{L}_2, \mathcal{L}_1\},
\end{gather*}
and along with our assumption that at least one of the terms $\frac{\partial G}{\partial \mathcal{L}_1}$, $\frac{\partial G}{\partial \mathcal{L}_2}$ is non-zero we see that $\{\mathcal{L}_1, \mathcal{L}_2\}= K\frac{\partial G}{\partial {\cal X}}$.
\end{proof}

\begin{Corollary}\label{lemma}
Keeping the same notations as of Definition~{\rm \ref{defn1}}, let $\mathcal{A}$ be a free abstract quadratic algebra. Then $ {G}(\mathcal{L}_1, \mathcal{L}_2, \mathcal{L}_3,\mathcal{L}_4 )$ depends non-trivially in at least two of the three generators $\mathcal{L}_1$, $\mathcal{L}_2$, $\mathcal{L}_4$.
\end{Corollary}

\begin{proof} It easily follows from the structure equations (\ref{eq}) that if $ {G}(\mathcal{L}_1, \mathcal{L}_2, \mathcal{L}_3,\mathcal{L}_4 )$ depends on at most one of the generators $\mathcal{L}_1$, $\mathcal{L}_2$, $\mathcal{L}_4$ then the center of~$\mathcal{A}$ contains a second order generator which is not a multiple of $\mathcal{H}$.
\end{proof}

Free abstract quadratic algebras appear as symmetry algebras of 2D degenerate free superintegrable systems. That is the main motivation for their study. For a~given~$ G$ there is a~1-parameter family of quadratic algebras, parametrized by~$K$. However, these algebras are isomorphic. It follows that the particular choice of nonzero $ K$ in the classif\/ication theory to follow is immaterial.

\begin{Note} If phase space generators ${\mathcal L}_1$, ${\mathcal L}_2$, ${\mathcal H}$, ${\mathcal X}$ satisfy a 4th order (in the momentum variables) relation ${ G}=0$ but the closure relations for $\{{\mathcal X},{\mathcal L}_1\}$, $\{{\mathcal X},{\mathcal L}_2\}$, $\{{\mathcal L}_1,{\mathcal L}_2\}$, are not satisf\/ied, then~$K$ is rational in the generators.

For quantum degenerate systems, knowledge of the Casimir relation $ G(L_1,L_2,H,X, \alpha)=0$ is suf\/f\/icient to determine the quadratic algebra, subject to the same conditions as the classical case. Details can be found in \cite{HKMS2015}.\end{Note}

\subsection{The symmetry group of a free abstract quadratic algebra}
In this section we determine the symmetry group of a free abstract quadratic algebra. Note that once we f\/ixed a generating set $\{\mathcal{L}_1, \mathcal{L}_2, \mathcal{L}_3,{\cal X} \}$, the Casimir $G$ is given by a symmetric $4\times 4$ matrix $B(G)$ def\/ined by
\begin{gather*}
 {G}(\mathcal{L}_1, \mathcal{L}_2, \mathcal{L}_3,\mathcal{L}_4)=\sum_{i,j=1}^{4}B({G})_{ij} \mathcal{L}_i\mathcal{L}_j,\end{gather*}
The pair $(B(G),K)$ (consisting of the Casimir and the multiplicative constant in the structure equations) is def\/ined only up to a constant. That is, the pair $\big(zB(G),z^{-1}K\big)$ for any nonzero~$z$ can also serve as a Casimir and a structure constant for the same quadratic algebra and the same basis. A free abstract quadratic algebra is completely determined by a generating set $ \{\mathcal{L}_1, \mathcal{L}_2, \mathcal{L}_3,{\cal X}  \}$ (satisfying the conditions of Def\/inition~\ref{defn1}) together with the pair $(B(G),K)$. As noted previously, the constant $K$ can always be normalized to~$1$.
Given two sets of generators $ \{\mathcal{L}_1, \mathcal{L}_2, \mathcal{L}_3,{\cal X}  \}$, $ \big\{ \widetilde{\mathcal{L}}_1,
 \widetilde{\mathcal{L}}_2, \widetilde{\mathcal{L}}_3, \widetilde{{\cal X}} \big\}$ of the same algebra~$\mathcal{A}$ such that:
 \begin{enumerate}\itemsep=0pt
 \item $2=\operatorname{deg}(\mathcal{L}_1)= \operatorname{deg}\big(\widetilde{\mathcal{L}}_1\big)=\operatorname{deg}(\mathcal{L}_2)= \operatorname{deg}(\widetilde{\mathcal{L}}_2)=\operatorname{deg}(\mathcal{L}_3)= \operatorname{deg}\big(\widetilde{\mathcal{L}}_3\big)$.
 \item $1=\operatorname{deg}(\mathcal{X})= \operatorname{deg}\big(\widetilde{\mathcal{X}}\big)$.
 \item ${\mathcal{L}}_3$ and $\widetilde{\mathcal{L}}_3$ are in the center of $\mathcal{A}$.
 \end{enumerate}
Then there is a 'change of basis matrix' $A$ of the form
\begin{gather}\label{eq21l}
A:=\left(\begin{matrix}
A_{1,1} &A_{1,2} &A_{1,3} & A_{1,4}&0 \\
A_{2,1} &A_{2,2} &A_{2,3}& A_{2,4}&0 \\
0 &0 &A_{3,3} & 0 &0\\
0 &0 &0& A_{4,4}&0\\
0 &0 &0&0 & A_{5,5}
\end{matrix}\right)\in {\rm GL}(5,\mathbb{C}),
\end{gather}
with $A_{4,4}=A_{5,5}^2$ such that for $\mathcal{L}:=(\mathcal{L}_1,\mathcal{L}_2,\mathcal{L}_3,\mathcal{L}_4,\mathcal{X})^{\rm t}$, $\widetilde{\mathcal{L}}:=\big(\widetilde{\mathcal{L}}_1,\widetilde{\mathcal{L}}_2,\widetilde{\mathcal{L}}_3,\widetilde{\mathcal{L}}_4, \widetilde{\mathcal{X}}\big)^{\rm t}$
 \begin{gather}\label{Tran'}
{\mathcal{L}}=A\widetilde{\mathcal{L}}.
\end{gather}

For $A$ as above we def\/ine
\begin{gather*}
 \widehat{A}=\left(\begin{matrix}
A_{1,1} &A_{1,2} &A_{1,3} & A_{1,4} \\
A_{2,1} &A_{2,2} &A_{2,3} & A_{2,4} \\
0 &0 &A_{3,3} &0\\
0 &0 &0 & A_{4,4}
\end{matrix}\right).
\end{gather*}
Besides linear change of basis we can also rescale the invariants $(B(G),K)$. The full symmetry group of a free abstract quadratic algebra is given by the collection of pairs $(A,z)$, $A$ as in equation~(\ref{eq21l}) and $z\in \mathbb{ C}^*$. We shall denote this group by $G_{\operatorname{degn}}$. The action of such $(A,z)\in G_{\operatorname{degn}}$ on the generators is given by equation (\ref{Tran'}) and on $(G,K)$ by
\begin{gather*}
 B(G')=z\widehat{A}^{\rm t} B({G}) \widehat{A},\\
 K'=z^{-1}K(A_{1,1}A_{2,2}-A_{1,2}A_{2,1})^{-1}A_{5,5}^{-1}, 
\end{gather*}
where $B({G})$ is the symmetric matrix representing $G$ with respect to $ \{\mathcal{L}_1, \mathcal{L}_2, \mathcal{L}_3,\mathcal{L}_4  \}$,
 and  $B(G')$ is the symmetric matrix representing $G':=(A,z)\cdot G$ with respect to $ \big\{\widetilde{\mathcal{L}}_1,\widetilde{\mathcal{L}}_2,\widetilde{\mathcal{L}}_3,\widetilde{\mathcal{L}}_4\big\}$. We shall call the subgroup of $G_{\operatorname{degn}}$ of all pairs of the form $(A,1)$ as the group of linear change of bases,we denote it by $G_{\operatorname{degn}}^{\operatorname{lin}}$. In the following we shall determine a canonical form for the pair $(B(G),K)$ of each abstract quadratic algebra. This means that we shall decide upon a unique representative from each orbit of $G_{\operatorname{degn}}$ in the space of all possible pairs $(B(G),K)$ that arise from an abstract quadratic algebra. Without loss of generality we can assume that $K=1$ and determine representative from each orbit of the subgroup of $G_{\operatorname{degn}}$ that f\/ixes the value of~$K$. Explicitly this group is given by
\begin{gather*}G_{\operatorname{degn}}^K:=\big\{(A,z)\in G_{\operatorname{degn}}\,|\,z=(A_{1,1}A_{2,2}-A_{1,2}A_{2,1})^{-1}A_{5,5}^{-1} \big\}.\end{gather*}

\subsection{The canonical form} \label{canonical form}
In this section after we identify some invariants of orbits of $(B(G),K)$ under the action of $G_{\operatorname{degn}}^{\operatorname{K}}$ we consider the reduction to a canonical form. That is, for each orbit we choose a representative which we call the canonical form $(B(G),K)$. For a given Casimir $(B(G),K)$ of a free abstract quadratic algebra $\mathcal{A}$ and a given $(A,z)\in G_{\operatorname{degn}}^{\operatorname{K}} $ we introduce the following notations:
\begin{gather*}
 \widehat{A}=\left(\begin{matrix}
A_{1,1} &A_{1,2} &A_{1,3} & A_{1,4} \\
A_{2,1} &A_{2,2} &A_{2,3} & A_{2,4} \\
0 &0 &A_{3,3} &0\\
0 &0 &0 & A_{4,4}
\end{matrix}\right)=\left(\begin{matrix}
r&s \\
 0&t
\end{matrix}\right),\\
 B({G})=\left(\begin{matrix}
b_{1,1} &b_{1,2} &b_{1,3} & b_{1,4} \\
b_{1,2} &b_{2,2} &b_{2,3} & b_{2,4} \\
b_{1,3} &b_{2,3} &b_{3,3} &b_{3,4}\\
b_{1,4} &b_{2,4} &b_{3,4} & b_{4,4}
\end{matrix}\right)=\left(\begin{matrix}
b&c \\
c^{\rm t}&d
\end{matrix}\right).
\end{gather*}
Direct calculation shows that
\begin{gather*}
 \widehat{A}^{\rm t}B({G}) \widehat{A}=\left(\begin{matrix}
r^{\rm t}br&r^{\rm t}(bs+ct) \\
s^{\rm t}br+tc^{\rm t}r&s^{\rm t}(bs+ct)+t(c^{\rm t}s+dt)
\end{matrix}\right).
\end{gather*}
Hence under the action $G_{\operatorname{degn}}^{\operatorname{K}} $ the rank of $B(G)$ and the rank of its upper left~2 by 2 block (which we denote by~$b$) are preserved. We use the theory of symmetric bilinear forms and f\/ind $(A,z)\in G_{\operatorname{degn}}^{\operatorname{K}} $ that cast $B(G)$ into exactly one of the following forms depending on $\operatorname{rank}(b)$.
\begin{gather*}
1. \quad B(G)=\left(\begin{matrix}
1 &0 &0 & 0 \\
0 &1&0 & 0 \\
0 &0 &b_{3,3} &b_{3,4}\\
0 &0 &b_{3,4} & b_{4,4}
\end{matrix}\right)
\qquad\text{if} \quad
\operatorname{rank}(b)=2,
\\ 
2. \quad B(G)=\left(\begin{matrix}
1 &0 &0 & 0 \\
0 &0&b_{2,3} &b_{2,4}\\
0 &b_{2,3} &b_{3,3} &b_{3,4}\\
0 &b_{2,4} &b_{3,4} & b_{4,4}
\end{matrix}\right)
\qquad \text{if} \quad
\operatorname{rank}(b)=1,
\\ 
3. \quad B(G)=\left(\begin{matrix}
0 &0 &b_{1,3} & b_{1,4} \\
0 &0&b_{2,3} & b_{2,4} \\
b_{1,3} &b_{2,3} &b_{3,3} &b_{3,4}\\
b_{1,4} &b_{2,4} &b_{3,4} & b_{4,4}
\end{matrix}\right)
\qquad \text{if} \quad
\operatorname{rank}(b)=0.
\end{gather*}
Below we analyze the dif\/ferent cases according to $\operatorname{rank}(b)$ we summarize the results of all cases in Table~\ref{table2}.

\subsection{The rank 2 case}\label{sec4.4}
If
\begin{gather*}
B( {G} )=\left(\begin{matrix}
1 &0 &0 & 0 \\
0 &1&0 & 0 \\
0 &0 &b_{3,3} &b_{3,4}\\
0 &0 &b_{4,3} & b_{4,4}
\end{matrix}\right) \qquad \text{and} \qquad (A,z)\cdot B(G)=\left(\begin{matrix}
1 &0 &0 & 0 \\
0 &1&0 & 0 \\
0 &0 &\widetilde{b}_{3,3} &\widetilde{b}_{3,4}\\
0 &0 &\widetilde{b}_{4,3} & \widetilde{b}_{4,4}
\end{matrix}\right),
\end{gather*}
for some $(A,z)\in G_{\operatorname{degn}}^{\operatorname{K}}$ then
\begin{gather*}
 \widehat{A}=\left(\begin{matrix}
A_{1,1} &A_{1,2} &0 & 0 \\
A_{2,1} &A_{2,2} &0 & 0 \\
0 &0 &A_{3,3} &0\\
0 &0 &0 & A_{4,4}
\end{matrix}\right), \qquad z=(A_{1,1}A_{2,2}-A_{1,2}A_{2,1})^{-1}A_{5,5}^{-1},
\end{gather*}
with
\begin{gather*}
\sqrt{z}\left(\begin{matrix}
A_{1,1} &A_{1,2} \\
A_{2,1} &A_{2,2}
\end{matrix}\right)\in {\rm O}_2(\mathbb{C}),\qquad A_{5,5}^{-1}=\pm 1,
\qquad
\left(\begin{matrix}
\widetilde{b}_{3,3} &\widetilde{b}_{3,4}\\
\widetilde{b}_{3,4} & \widetilde{b}_{4,4}
\end{matrix}\right)= \left(\begin{matrix}
zA_{3,3}^2b_{3,3} &zA_{3,3}b_{3,4}\\
zA_{3,3}b_{3,4} & zb_{4,4}
\end{matrix}\right).
\end{gather*}
So if $b_{3,3} \neq 0$ we def\/ine the canonical
form of $B( {G} )$ to be
\begin{gather*}
B( {G} )^{21}(b_{3,4}, b_{4,4})=\left(\begin{matrix}
1 &0 &0 & 0 \\
0 &1&0 & 0 \\
0 &0 &1 &b_{3,4}\\
0 &0 &b_{3,4} & b_{4,4}
\end{matrix}\right),
\end{gather*}
with $b_{3,4}\in \{0,1\}$, $b_{4,4}\in {\mathcal C}$. If $b_{3,3}= 0$ we def\/ine the canonical form of $B( {G} )$ to be
\begin{gather*}
B( {G} )^{22}(b_{3,4}, b_{4,4})=\left(\begin{matrix}
1 &0 &0 & 0 \\
0 &1&0 & 0 \\
0 &0 &0 &b_{3,4}\\
0 &0 &b_{3,4} & b_{4,4}
\end{matrix}\right),
\end{gather*}
with $b_{3,4},b_{4,4} \in \{0,1\}$.

\subsection{The rank 1 case}\label{sec4.5}
If
\begin{gather*}
B( G )=\left(\begin{matrix}
1 &0 &0 & 0 \\
0 &0&b_{2,3} &b_{2,4}\\
0 &b_{2,3} &b_{3,3} &b_{3,4}\\
0 &b_{2,4} &b_{3,4} & b_{4,4}
\end{matrix}\right)\qquad \text{and} \qquad (A,z)\cdot B(G)=\left(\begin{matrix}
1 &0 &0 & 0 \\
0 &0&\widetilde{b}_{2,3} &\widetilde{b}_{2,4}\\
0 &\widetilde{b}_{2,3} &\widetilde{b}_{3,3} &\widetilde{b}_{3,4}\\
0 &\widetilde{b}_{2,4}&\widetilde{b}_{3,4} & \widetilde{b}_{4,4}
\end{matrix}\right),
\end{gather*}
for some $(A,z)\in G_{\operatorname{degn}}^{\operatorname{K}}$ then
\begin{gather*}
 \widehat{A}=\left(\begin{matrix}
\pm {z}^{-1/2} &0 &\mp {z}^{1/2}b_{2,3}A_{3,3}A_{2,1} & \mp {z}^{-1/2}b_{2,4}A_{4,4}A_{2,1} \\
A_{2,1} &A_{2,2} &A_{2,3}& A_{2,4} \\
0 &0 &A_{3,3} &0\\
0 &0 &0 & A_{4,4}
\end{matrix}\right),\qquad
  z=A_{2,2}^{-2}A_{4,4}^{-1},
\end{gather*}
and
\begin{gather}\label{equ46}
\widetilde{b}_{2,3}=b_{2,3}A_{2,2}^{-1}A_{3,3} A_{4,4}^{-1},\\\label{equ47}
 \widetilde{b}_{2,4}= b_{2,4}A_{2,2}^{-1},\\
 \widetilde{b}_{3,3}= b_{2,3}^2A_{2,1}^2A_{2,2}^{-4}A_{3,3}^2A_{4,4}^{-2}+2b_{2,3}A_{2,3}A_{2,2}^{-2}A_{3,3}A_{4,4}^{-1}+b_{3,3}A_{2,2}^{-2}A_{3,3}^2A_{4,4}^{-1},\nonumber \\
 \widetilde{b}_{3,4}=b_{2,3}b_{2,4}A_{2,1}^2A_{2,2}^{-4}A_{3,3}A_{4,4}^{-1}+b_{2,3}A_{2,4}A_{2,2}^{-2}A_{3,3}A_{4,4}^{-1}
+b_{2,4}A_{2,3}A_{2,2}^{-2}+b_{3,4}A_{2,2}^{-2}A_{3,3}, \nonumber\\
 \widetilde{b}_{4,4}= b_{2,4}^2A_{2,1}^2A_{2,2}^{-4}+2b_{2,4}A_{2,4}A_{2,2}^{-2}+b_{4,4}A_{2,2}^{-2}A_{4,4}.\nonumber
\end{gather}
From (\ref{equ46}) we see that $\widetilde{b}_{2,3}=0$ if and only if ${b}_{2,3}=0$, similarly from (\ref{equ47}) we see that $\widetilde{b}_{2,4}=0$ if and only if ${b}_{2,4}=0$. Hence we continue our analysis according to the vanishing of $b_{2,3}$ and~$b_{2,4}$.

\subsubsection[$\operatorname{rank}(b)=1$, $b_{2,3}=0$ and $b_{2,4}=0$]{$\boldsymbol{\operatorname{rank}(b)=1}$, $\boldsymbol{b_{2,3}=0}$ and $\boldsymbol{b_{2,4}=0}$}
If
\begin{gather*}
B( G )=\left(\begin{matrix}
1 &0 &0 & 0 \\
0 &0&0 & 0 \\
0 &0 &b_{3,3} &b_{3,4}\\
0 &0 &b_{3,4} &b _{4,4}
\end{matrix}\right) \qquad \text{and} \qquad (A,z)\cdot B(G)=\left(\begin{matrix}
1 &0 &0 & 0 \\
0 &0&0 & 0 \\
0 &0 &\widetilde{b}_{3,3} &\widetilde{b}_{3,4}\\
0 &0 &\widetilde{b}_{3,4} & \widetilde{b}_{4,4}
\end{matrix}\right),
\end{gather*}
for some $(A,z)\in G_{\operatorname{degn}}^{\operatorname{K}}$ then
\begin{gather*}
 \widehat{A}=\left(\begin{matrix}
\pm {z}^{-1/2} &0 &0 & 0 \\
A_{2,1} &A_{2,2} &A_{2,3}& A_{2,4} \\
0 &0 &A_{3,3} &0\\
0 &0 &0 & A_{4,4}
\end{matrix}\right), \qquad z=A_{2,2}^{-2}A_{4,4}^{-1},
\end{gather*}
and
\begin{gather*}
\left(\begin{matrix}
\widetilde{b}_{3,3} &\widetilde{b}_{3,4}\\
\widetilde{b}_{3,4} & \widetilde{b}_{4,4}
\end{matrix}\right)= \left(\begin{matrix}
A_{2,2}^{-2}A_{4,4}^{-1}A_{3,3}^2b_{3,3} &A_{2,2}^{-2}A_{4,4}^{-1}A_{3,3}A_{4,4}b_{3,4}\\
A_{2,2}^{-2}A_{4,4}^{-1}A_{3,3}A_{4,4}b_{3,4} &A_{2,2}^{-2}A_{4,4}b_{4,4}
\end{matrix}\right).
\end{gather*}
We def\/ine the canonical form to be
\begin{gather*}
B( G )^{11}(b_{3,3},b_{3,4},b_{4,4})=\left(\begin{matrix}
1 &0 &0 & 0 \\
0 &0&0 & 0 \\
0 &0 &b_{3,3} &b_{3,4}\\
0 &0 &b_{3,4} & b_{4,4}
\end{matrix}\right),
\end{gather*}
with
$b_{3,3},b_{3,4},b_{4,4}\in \{0,1\}$.

\subsubsection[$\operatorname{rank}(b)=1$, $b_{2,3}\neq 0$ and $b_{2,4}\neq 0$]{$\boldsymbol{\operatorname{rank}(b)=1}$, $\boldsymbol{b_{2,3}\neq 0}$ and $\boldsymbol{b_{2,4}\neq 0}$}
In this case we can act with $(A,z)\in G_{\operatorname{degn}}^{\operatorname{K}}$
to obtain
\begin{gather*}
B(G)^{15}(b_{3,4})=\left(\begin{matrix}
1 &0 &0 & 0 \\
0 &0&1 &1\\
0 &1 &0 &b_{3,4}\\
0 &1&b_{3,4} & 0
\end{matrix}\right),
\end{gather*}
for a unique $b_{3,4}\in \{0,1\}$.
We def\/ine this form to be the canonical form in this case.

\subsubsection[$\operatorname{rank}(b)=1$, $ b_{2,3}\neq 0$ and $b_{2,4}= 0$]{$\boldsymbol{\operatorname{rank}(b)=1}$, $\boldsymbol{b_{2,3}\neq 0}$ and $\boldsymbol{b_{2,4}= 0}$}
In this case we can act with $(A,z)\in G_{\operatorname{degn}}^{\operatorname{K}}$ to obtain
\begin{gather*}
B(G)^{16}(b_{4,4})=\left(\begin{matrix}
1 &0 &0 & 0 \\
0 &0&1 &0\\
0 &1 &0 &0\\
0 &0&0 & b_{4,4}
\end{matrix}\right),
\end{gather*}
with unique $b_{4,4}\in \{0,1\}$ which is def\/ined to be the canonical form in this case.

\subsubsection[$\operatorname{rank}(b)=1$, $b_{2,3}= 0$ and $b_{2,4}\neq 0$]{$\boldsymbol{\operatorname{rank}(b)=1}$, $\boldsymbol{b_{2,3}= 0}$ and $\boldsymbol{ b_{2,4}\neq 0}$}

In this case we can act with $(A,z)\in G_{\operatorname{degn}}^{\operatorname{K}}$ to obtain
\begin{gather*}
B(G)^{17}(b_{3,3})=\left(\begin{matrix}
1 &0 &0 & 0 \\
0 &0&0 &1\\
0 &0 &b_{3,3} &0\\
0 &1&0 & 0
\end{matrix}\right),
\end{gather*}
with unique $b_{3,3}\in \{0,1\}$ which is def\/ined to be the canonical form in this case.

\subsection{The rank 0 case}\label{sec4.6}
If
\begin{gather*}
B( G )=\left(\begin{matrix}
0 &0 &b_{1,3} &b_{1,4} \\
0 &0&b_{2,3} &b_{2,4}\\
b_{1,3}&b_{2,3} &b_{3,3} &b_{3,4}\\
b_{1,4}&b_{2,4} &b_{3,4} & b_{4,4}
\end{matrix}\right)\qquad \text{and} \qquad (A,z)\cdot B(G)=\left(\begin{matrix}
0 &0 &\widetilde{b}_{1,3} & \widetilde{b}_{1,4} \\
0 &0&\widetilde{b}_{2,3} &\widetilde{b}_{2,4}\\
\widetilde{b}_{1,3} &\widetilde{b}_{2,3} &\widetilde{b}_{3,3} &\widetilde{b}_{3,4}\\
 \widetilde{b}_{1,4}&\widetilde{b}_{2,4}&\widetilde{b}_{3,4} & \widetilde{b}_{4,4}
\end{matrix}\right),
\end{gather*}
for some $(A,z)\in G_{\operatorname{degn}}^{\operatorname{K}}$ then
\begin{gather*}
 \left(\begin{matrix}
\widetilde{b}_{1,3} &\widetilde{b}_{1,4}\\
\widetilde{b}_{2,4} & \widetilde{b}_{2,4}
\end{matrix}\right) = z\left(\begin{matrix}
 A_{1,1} & A_{2,1}\\
 A_{1,2} & A_{2,2}
\end{matrix}\right)\left(\begin{matrix}
 b_{1,3} & b_{1,4}\\
 b_{2,3} & b_{2,4}
\end{matrix}\right)\left(\begin{matrix}
 A_{3,3} & 0\\
 0 & A_{4,4}
\end{matrix}\right),
\end{gather*}
and
\begin{gather*}
\left(\begin{matrix}
\widetilde{b}_{3,3} &\widetilde{b}_{3,4}\\
\widetilde{b}_{4,3} & \widetilde{b}_{4,4}
\end{matrix}\right)= z\left(\begin{matrix}
 A_{1,3} & A_{2,3}\\
 A_{1,4} & A_{2,4}
\end{matrix}\right)\left(\begin{matrix}
 b_{1,3} & b_{1,4}\\
 b_{2,3} & b_{2,4}
\end{matrix}\right)\left(\begin{matrix}
 A_{3,3} & 0\\
 0 & A_{4,4}
\end{matrix}\right) \\
\hphantom{\left(\begin{matrix}
\widetilde{b}_{3,3} &\widetilde{b}_{3,4}\\
\widetilde{b}_{4,3} & \widetilde{b}_{4,4}
\end{matrix}\right)=}{} +
 z\left(\begin{matrix}
 A_{3,3} & 0\\
 0 & A_{4,4}
\end{matrix}\right)\left(\begin{matrix}
 b_{1,3} & b_{2,3}\\
 b_{1,4} & b_{2,4}
\end{matrix}\right)\left(\begin{matrix}
 A_{1,3} & A_{1,4}\\
 A_{2,3} & A_{2,4}
\end{matrix}\right)  \\
\hphantom{\left(\begin{matrix}
\widetilde{b}_{3,3} &\widetilde{b}_{3,4}\\
\widetilde{b}_{4,3} & \widetilde{b}_{4,4}
\end{matrix}\right)=}{}+
 z\left(\begin{matrix}
 A_{3,3} & 0\\
 0 & A_{4,4}
\end{matrix}\right)\left(\begin{matrix}
 b_{3,3} & b_{3,4}\\
 b_{3,4} & b_{4,4}
\end{matrix}\right)\left(\begin{matrix}
 A_{3,3} & 0\\
 0 & A_{4,4}
\end{matrix}\right).
\end{gather*}
Hence in the case of $b=0$, the rank of $c=\left(\begin{smallmatrix}
 b_{1,3} & b_{1,4}\\
 b_{2,3} & b_{2,4}
\end{smallmatrix}\right)$ is invariant under the action of $G_{\operatorname{degn}}^{\operatorname{K}}$.
We continue our analysis according to the value of this rank. By Corollary~\ref{lemma} $c\neq 0$.

\subsubsection[$\operatorname{rank}(b)=0$, $\operatorname{rank}(c)=1$]{$\boldsymbol{\operatorname{rank}(b)=0}$, $\boldsymbol{\operatorname{rank}(c)=1}$}
In this case we have the following forms
\begin{gather*}\nonumber
\left(\begin{matrix}
0 &0 &0 & 0 \\
0 &0&1 & 1 \\
0 &1 &b_{3,3} &b_{3,4}\\
0 &1 &b_{3,4} & b_{4,4}
\end{matrix}\right),\qquad
\left(\begin{matrix}
0 &0 &0 & 0 \\
0 &0&1 & 0 \\
0 &1 &b_{3,3} &b_{3,4}\\
0 &0 &b_{3,4} & b_{4,4}
\end{matrix}\right),\qquad
 \left(\begin{matrix}
0 &0 &0 & 0 \\
0 &0&0 & 1 \\
0 &0 &b_{3,3} &b_{3,4}\\
0 &1 &b_{3,4} & b_{4,4}
\end{matrix}\right).
\end{gather*}
The f\/irst case can be reduced to the canonical form
\begin{gather*}
B(G)^{05}=\left(\begin{matrix}
0 &0 &0 & 0 \\
0 &0&1 &1\\
0 &1 &0 &b_{3,4}\\
0 &1&b_{3,4} & 0
\end{matrix}\right),
\end{gather*}
for a unique $b_{3,4}\in \{0,1\}$. The second case can be reduced to the canonical form
\begin{gather*}
B(G)^{06}=\left(\begin{matrix}
0 &0 &0 & 0 \\
0 &0&1 & 0 \\
0 &1 &0 &0\\
0 &0 &0 & b_{4,4}
\end{matrix}\right),
\end{gather*}
for a unique $b_{4,4}\in\{0,1\}$. The case with $b_{4,4}=0$ is not possible due to Corollary~\ref{lemma}.
The third case can be reduced to the following canonical form
\begin{gather*}
B(G)^{07}(b_{3,3})=\left(\begin{matrix}
0 &0 &0 & 0 \\
0 &0&0 & 1 \\
0 &0 &b_{3,3} &0\\
0 &1 &0 & 0
\end{matrix}\right),
\end{gather*}
for a unique $b_{3,3}\in \{0,1\}$.

\subsubsection[$\operatorname{rank}(b)=0$, $\operatorname{rank}(c)=2$]{$\boldsymbol{\operatorname{rank}(b)=0}$, $\boldsymbol{\operatorname{rank}(c)=2}$}
In this case we only have the form
\begin{gather}
\left(\begin{matrix}
0 &0 &1 & 0 \\
0 &0&0 & 1 \\
1 &0 &b_{3,3} &b_{3,4}\\
0 &1 &b_{3,4} & b_{4,4}
\end{matrix}\right),
\end{gather}
which can be reduced to the canonical form
\begin{gather*}
B(G)^{08}=\left(\begin{matrix}
0 &0 &1 & 0 \\
0 &0&0 & 1 \\
1 &0 &0 &0\\
0 &1 &0 & 0
\end{matrix}\right).
\end{gather*}

\begin{table}[h]\centering
 \begin{tabular}{|l |l | l | p{8.5cm}|}
 \hline \multicolumn{4}{|c|}{degenerate quadratic algebras}\\
 \hline $\#$& $\operatorname{rank}(b)$& invariant form & canonical form
\\ \hline
1& 2& $B( {G} )^{21}(b_{34},b_{44})$ &$\mathcal{L}_1^2+ \mathcal{L}_2^2+ \mathcal{H}^2+b_{44} {\cal X}^4+2b_{34}\mathcal{H}{\cal X}^2$,\tsep{2pt} $ b_{34}\in \{0,1\}$, $b_{44} \in \mathbb{C} $\bsep{2pt}\\ \hline
2& 2& $B( {G} )^{22}(b_{34},b_{44})$ & $\mathcal{L}_1^2+ \mathcal{L}_2^2+b_{44} {\cal X}^4+2b_{34}\mathcal{H}{\cal X}^2$,\tsep{2pt}  $b_{3,4},b_{44}\in \{0,1\}$\bsep{2pt}\\ \hline
3& 1& $B( {G} )^{11}(b_{33},b_{34},b_{44})$ & $\mathcal{L}_1^2+ b_{33}\mathcal{H}^2+ b_{44}{\cal X}^4+2b_{34}\mathcal{H}{\cal X}^2$,\tsep{2pt} $ b_{33},b_{34},b_{44}\in \{0,1\} $, $b_{34}+b_{44}\neq0$\bsep{2pt}\\ \hline
4& 1& $B( {G} )^{15}(b_{34})$ & $\mathcal{L}_1^2+2\mathcal{L}_2\mathcal{H} +2\mathcal{L}_2{\cal X}^2+2b_{34}\mathcal{H}{\cal X}^2$, $ b_{34} \in \{0,1\} $
\tsep{2pt}\bsep{2pt}\\ \hline
5& 1& $B( {G} )^{16}(b_{44})$ & $\mathcal{L}_1^2+2\mathcal{L}_2\mathcal{H} + b_{44} {\cal X}^4$, $b_{44}\in \{0,1\}$
\tsep{2pt}\bsep{2pt}\\ \hline
6& 1& $B(G)^{17}(b_{33})$ & $\mathcal{L}_1^2+2\mathcal{L}_2\mathcal{X}^2 +b_{33}{\mathcal{H}}^2$, $b_{33}\in \{0,1\}$
\tsep{2pt}\bsep{2pt}\\ \hline
7& 0& $B( {G} )^{05}(b_{34})$ & $2\mathcal{L}_2\mathcal{H}+2\mathcal{L}_2{\cal X}^2 +2b_{34}\mathcal{H}{\cal X}^2$, $ b_{34} \in \{0,1\}$
\tsep{2pt}\bsep{2pt}\\ \hline
8& 0& $B( {G} )^{06}$ & $2\mathcal{L}_2\mathcal{H} +{\cal X}^4$
\tsep{2pt}\bsep{2pt}\\ \hline
9&0& $B(G)^{07}(b_{33})$ & $2\mathcal{L}_2{\cal X}^2+b_{33}\mathcal{H}^2$, $b_{33}\in \{0,1\}$
\tsep{2pt}\bsep{2pt}\\ \hline
10 & 0& $B( {G} )^{08}$ & $2\mathcal{L}_1\mathcal{H}+ 2\mathcal{L}_2{\cal X}^2$
\tsep{2pt}\bsep{2pt}\\ \hline
\end{tabular}
\caption{List of canonical forms of free abstract degenerate quadratic algebras. The canonical forms $B(G)^{ab}$ are given explicitly in Sections \ref{sec4.4}, \ref{sec4.5}, \ref{sec4.6} above.}\label{table2}
\end{table}

\subsection{Comparison of geometric and abstract free degenerate quadratic algebras}\label{Comparison}

We examine the entries in Table \ref{table2} to determine which can be realized as a classical 2D degenerate superintegrable system or a classical Lie algebra with linearly
independent generators.
There is a close relationship between the canonical forms of free abstract degenerate quadratic algebras and St\"ackel equivalence classes of degenerate
superintegrable systems. We demonstrate this by treating one example in detail. The superintegrable system $S3$, with degenerate potential,
can be def\/ined by
\begin{gather*}{\mathcal G}= {\mathcal L}_1^2+{\mathcal L}_2^2-{\mathcal L}_1{\mathcal H}+{\mathcal L}_1{\mathcal X}^2+a_1{\mathcal X}^2+(a_1+a_2){\mathcal L}_1=0,\end{gather*}
where the $a_j$ are the parameters in the potential. To perform a general St\"ackel transform of this system with nonsingular transform
matrix $C=(c_{jk})$: 1) we set
$a_j=\sum\limits_{k=1}^2 c_{jk}b_k$, $k=1,2$ where the $b_k$ are the new parameters, 2) we make the replacements
${\mathcal H}\to -b_2$, $b_2\to -{\mathcal H}$
and 3) we then set all parameters $b_j=0$ to determine the free degenerate quadratic algebra.
The result is
\begin{gather*}\label{Acanon}[A]\colon \quad {\mathcal G}={\mathcal L}_1^2+{\mathcal L}_2^2+{\mathcal L}_1{\mathcal X}^2-c_{12}{\mathcal H}{\mathcal X}^2
 -(c_{12}+c_{22}){\mathcal H}{\mathcal L}_1=0,\end{gather*}
 where $|c_{12}|+|c_{22}|>0$.
The canonical forms in Table~\ref{table2} associated with the equivalence class $[A]$ are
$1$:~$b_{44}=1$ and $2$:~$b_{44}=1$.

The superintegrable system $E6$, with degenerate potential, can be def\/ined by
\begin{gather*} {\mathcal G}={\mathcal L}_1^2-{\mathcal L}_2{\mathcal H}+{\mathcal L}_2{\mathcal X}^2+a_1{\mathcal X}^2+a_2{\mathcal L}_2=0.\end{gather*}
 Going through the same procedure as above, we obtain the equivalence class
 \begin{gather*}
 [B]\colon \quad {\mathcal G}={\mathcal L}_1^2+{\mathcal L}_2{\mathcal X}^2-c_{12}{\mathcal H}{\mathcal X}^2
 -c_{22}{\mathcal H}{\mathcal L}_2=0,\end{gather*} where $|c_{12}|+|c_{22}|>0$.
 The canonical form associated with this equivalence class is $4$: all cases.

 The superintegrable system $E3$, with degenerate potential, can be def\/ined by
\begin{gather*} {\mathcal G}={\mathcal L}_1^2+{\mathcal L}_2^2-{\mathcal L}_1{\mathcal H}+a_2\big({\mathcal X}^2+a_2\big){\mathcal L}_1.\end{gather*}
 The equivalence class is
 \begin{gather*}
 [C]\colon \quad {\mathcal G}={\mathcal L}_1^2+{\mathcal L}_2^2-c_{12}{\mathcal H}{\mathcal X}^2-c_{22}{\mathcal H}{\mathcal L}_1=0,
 \end{gather*} where $|c_{12}|+|c_{22}|>0$.
 The canonical forms associated with this equivalence class are $1$: $b_{44}=0$, and $2$:~$b_{44}=0$.

The superintegrable system $E5$ can be def\/ined by
\begin{gather*}{\mathcal G}= {\mathcal L}_1^2+{\mathcal X}^4-{\mathcal H}{\mathcal X}^2+a_1{\mathcal L}_2+a_2{\mathcal X}^2=0.\end{gather*}
The equivalence class is
\begin{gather*}
[D]\colon \quad {\mathcal G}={\mathcal L}_1^2+{\mathcal X}^4-c_{22}{\mathcal H}{\mathcal X}^2
-c_{12}{\mathcal H}{\mathcal L}_2=0,\end{gather*}
where $|c_{12}|+|c_{22}|>0$. The canonical forms associated with this equivalence class are
 $3$:~$b_{44}=1$, $b_{34}=0$ and $5$:~$b_{44}=1$.

The superintegrable system $E14$ can be def\/ined by
\begin{gather*}{\mathcal G}= -{\mathcal L}_1^2 -{\mathcal L}_2{\mathcal X}^2 +a_1{\mathcal H}-a_1a_2=0.\end{gather*}
 The equivalence class is
\begin{gather*}
[E]\colon \quad {\mathcal G}^2=-{\mathcal L}_1^2-{\mathcal L}_2{\mathcal X}^2-c_{12}c_{22}{\mathcal H}^2=0,\end{gather*}
where $|c_{12}|+|c_{22}|>0$.
 The canonical forms associated with $[E]$ are $6$: all cases.

The superintegrable system $E4$ can be def\/ined by
\begin{gather*} {\mathcal G}= {\mathcal H}^2+{\mathcal X}^4+2{\mathcal H}{\mathcal X}^2-4{\mathcal L}_2{\mathcal X}^2
-4ia_1{\mathcal L}_1-2a_2{\mathcal X}^2-2a_2{\mathcal H}+a_2^2=0.\end{gather*}
The equivalence class is
\begin{gather*}
 [F]\colon \quad {\mathcal G}={\mathcal X}^4-4{\mathcal L}_2{\mathcal X}^2-4ic_{12}{\mathcal H}{\mathcal L}_1
+2c_{22}{\mathcal H}{\mathcal X}^2
+c_{22}^2{\mathcal H}^2=0,\end{gather*}
where $|c_{12}|+|c_{22}|>0$.
The canonical forms associated with $[F]$ are
 $9$:~$b_{33}=1$ and $10$: all cases.

{\bf Heisenberg systems:}
In addition there are systems that can be obtained from the degenerate geometric systems above by contractions from ${\mathfrak{so}}(4,\mathbb{C})$ to $\mathfrak{e}(3,\mathbb{C})$.
These are not B\^ocher contractions and the contracted systems are not superintegrable, because the Hamiltonians become singular. However, they do
form quadratic algebras and several have the interpretation of time-dependent Schr\"odinger equations in 2D spacetime, so we also consider them geometrical.
Some of these were classif\/ied in~\cite{KM2014} where they were called Heisenberg systems since they appeared in quadratic algebras formed
from 2nd order elements in the Heisenberg algebra with generators ${\cal M}_1=p_x$, ${\cal M}_2=xp_y$, ${\cal E}=p_y$, where ${\cal E}^2={\cal H}$.
The systems are all of type~4. We will devote a future paper to their study. The possible canonical forms
are $3$:~$b_{33}=b_{44}=0$, $b_{34}=1$ and $5$:~$b_{44}=0$.

\begin{table}[t]\centering
\begin{tabular}{|l|l|l|}
\hline rank & \multicolumn{2}{c}{canonical form} \vline
\\ \hline 2 & $ 1$: all cases & $2$: all cases
\\ \hline 1 & $3$: all cases except & $4$: all cases
\\ \hline
& $(b_{44},b_{34}=1,b_{33}=0,1)$, $(b_{44}=0,b_{34},b_{33}=1)$ &\\ \hline
 1& $5$: all cases & $6$: all cases \\ \hline
0&$7$: no & $8$: missing $2{\mathcal L}_2{\mathcal H}+{\mathcal X}^4=0$ \\ \hline
0& $9$: missing ${\mathcal L}_2{\mathcal X}^2=0$ &$10$: all cases\\ \hline
\end{tabular}
\caption{Matching of geometric with abstract quadratic algebras.}\label{tablea}
\end{table}

All these results relating geometric systems to abstract systems are summarized in Table~\ref{tablea}. We see that every abstract quadratic algebra is isomorphic to a quadratic algebra corresponding to a superintegrable system, with just 5 exceptions.

\begin{Theorem} Every free quadratic algebra realizable by functions on $4$-dimensional phase space with grading the order of polynomials in the momenta is isomorphic to a free quadratic algebra of a superintegrable system. \end{Theorem}

Proof: We show that the 5 exceptional free quadratic algebras cannot be realized in phase space. We assume in each case that the algebra is realizable in terms of functions on phase space and obtain a contradiction.
\begin{enumerate}\itemsep=0pt
\item Case 3: ${\cal L}_1^2 +\big({\cal H}+{\cal X}^2\big)^2=0$.
 We can factor ${\cal G}$ as
$({\cal H}+{\cal X}^2+i{\cal L}_1)({\cal H}+{\cal X}^2-i{\cal L}_1)=0$. This is possible only if one of the factors vanishes; hence the generators are linearly dependent. Impossible!
\item Case 3: ${\cal L}_1^2 +2{\cal H}{\cal X}^2 +{\cal X}^4=0$.
Here, ${\cal L}_1^2=-{\cal X}^2(2{\cal H}+{\cal X}^2)$ so ${\cal L}_1={\cal X}{\cal Y}$ where ${\cal Y}$ is a 1st order constant of the motion. Thus ${\cal Y}$ must be a multiple of ${\cal X}$ and the generators are linearly dependent. Impossible!
\item Case 3: $ {\cal L}_1^2 +{\cal H}^2+2{\cal H}{\cal X}^2 =0$.
Here, ${\cal L}_1^2 =-{\cal H}({\cal H}+2{\cal X}^2)$. If ${\cal L}_1$ doesn't factor then it must be a multiple of ${\cal H}$ and of ${\cal H}+2{\cal X}^2$ simultaneously. Impossible! Suppose then that ${\cal L}_1={\cal Y}{\cal Z}$, a product of two 1st order factors. If ${\cal Z}$ is a multiple of ${\cal Y}$ then ${\cal Y}$ is a constant of the motion, hence proportional to~${\cal X}$. Impossible! Thus ${\cal Y},{\cal Z}$ must be distinct linear factors. If ${\cal H}$ is divisible by each factor then ${\cal L}_1$, ${\cal H}$ are linearly dependent. Impossible! So~$\cal H$ must be divisible by the square of a single factor. By renormalizing $\cal Y$ and $\cal Z$ we can assume ${\cal H}= {\cal Y}^2$. Thus~${\cal Y}$ is a 1st order constant of the motion, necessarily proportional to~${\cal X}$. We conclude that ${\cal H}\sim {\cal X}^2$. Impossible!
\item Case 8: $2{\cal L}_2{\cal H}+{\cal X}^4=0$.
Since ${\cal X}$ is 1st order, and $2{\cal L}_2{\cal H}=-{\cal X}^4$ both ${\cal H}$ and ${\cal L}_2$ must be divisible by ${\cal X}^2$. Thus each is a perfect square of a 1st order symmetry, necessarily a~scalar multiple of $\cal X$. Hence the 2nd order generators are linearly dependent. Impossible!
\item Case 9: ${\cal L}_2{\cal X}=0$.
Here at least one of the generators must vanish. Impossible.!
\end{enumerate}

Thus we have shown that the only free degenerate quadratic algebras that can be constructed in phase space are those that arise from superintegrability. The remaining 5 abstract systems must lie in dif\/ferent graded Poisson algebras.

\section[Classif\/ication of contractions of free abstract quadratic degenerate 2D superintegrable systems on constant curvature spaces and Darboux spaces]{Classif\/ication of contractions of free abstract quadratic\\ degenerate 2D superintegrable systems on constant\\ curvature spaces and Darboux spaces}\label{section5}

In this section we def\/ine contractions between free abstract quadratic algebras. Then we list the canonical forms of the Casimirs of free abstract quadratic algebras that arises as the symmetry algebras of degenerate 2D free superintegrable systems on constant curvature spaces or Darboux spaces. Finally using the canonical forms, we classify all possible contractions relating free abstract quadratic algebras of degenerate 2D superintegrable systems on constant curvature spaces or Darboux spaces.

\subsection{Contraction of a free abstract quadratic algebra}
\begin{Definition}
Let $\mathcal{A}$, $\mathcal{B}$ be free abstract quadratic algebras with Casimirs $(B({G}_\mathcal{A}),K_\mathcal{A})$ and
$(B({G}_\mathcal{B}),K_\mathcal{B})$, respectively.
Suppose that there is a continuous map $\epsilon \longmapsto (A_{\epsilon},z_{\epsilon})$ from some punctured neighborhood of $0$ in $\C^*$,
the nonzero complex numbers, into $G_{\operatorname{degn}}$ and such that
$\lim\limits_{\epsilon \longrightarrow 0} (A_{\epsilon},z_{\epsilon})\cdot (b({G}_\mathcal{A}),K_\mathcal{A})=(B({G}_\mathcal{B}),K_\mathcal{B})$.
Then we say that $\mathcal{B}$ is a contraction of $\mathcal{A}$.\end{Definition}

The meaning of the def\/inition is that in any generating set $\left\{\mathcal{L}_1, \mathcal{L}_2, \mathcal{L}_3,\mathcal{L}_4 \right\}$
of $\mathcal{A}$ (satisfying the same assumptions as before) the corresponding matrix $B({G}_\mathcal{A})$ and ${K}_\mathcal{A}$ satisfy
\begin{gather}\label{e67}
\lim_{\epsilon \longrightarrow 0}z_{\epsilon}\widehat{A_{\epsilon}}^{\rm t} B({G}_\mathcal{A}) \widehat{A_{\epsilon}}=B({G}_\mathcal{B}), \\ \label{e68}
\lim_{\epsilon \longrightarrow 0}z_{\epsilon}^{-1}K_\mathcal{A}\left((A_{\epsilon})_{1,1}
(A_{\epsilon})_{2,2}-(A_{\epsilon})_{1,2}(A_{\epsilon})_{2,1}\right)^{-1}(A_{\epsilon})_{5,5}^{-1}=K_\mathcal{B},
\end{gather}
where $(B({G}_\mathcal{B}),K_\mathcal{B})$ is a realization of the Casimir of $\mathcal{B}$ in some generating set. In the classif\/ication below we are using a more ref\/ined class of contractions. We shall call these contractions algebraic. By def\/inition an \textit{algebraic contraction} is a contraction that can be realized via a map $\epsilon \longmapsto (A_{\epsilon},z_{\epsilon})$ from some punctured neighborhood of $0$ in $\C^*$ into $G_{\operatorname{degn}}$ such that $z_{\epsilon}$ as well as the entries of $A_{\epsilon}$ are rational functions in $\epsilon $.

\begin{Proposition}\label{p1}
$\mathcal{B}$ is a contraction of $\mathcal{A}$ if and only if there is a continuous map $\epsilon \longmapsto (A_{\epsilon},z_{\epsilon})$
from some punctured neighborhood of $0$ in $\C^*$ into $G_{\operatorname{degn}}$ such that
\begin{gather}\label{690}
 \lim_{\epsilon \longrightarrow 0}z_{\epsilon}\widehat{A_{\epsilon}}^{\rm t} B({G}_\mathcal{A})_{\rm can} \widehat{A_{\epsilon}}=B({G}_\mathcal{B})_{\rm can}, \\ \label{700}
 \lim_{\epsilon \longrightarrow 0}z_{\epsilon}^{-1}\left((A_{\epsilon})_{1,1}(A_{\epsilon})_{2,2}-(A_{\epsilon})_{1,2}(A_{\epsilon})_{2,1},
\right)^{-1}(A_{\epsilon})_{5,5}^{-1}=1,
\end{gather}
where $(B({G}_\mathcal{A})_{\rm can},1)$ and $(B({G}_\mathcal{B})_{\rm can},1)$ are the canonical forms of the Casimirs of~$\mathcal{A}$ and $\mathcal{B}$ respectively.
\end{Proposition}

 \begin{proof} Obviously if equations (\ref{690})--(\ref{700}) hold then $\mathcal{B}$ is a contraction of $\mathcal{A}$. For the other direction assume that equations (\ref{e67})--(\ref{e68}) hold. We can further assume that $(B({G}_\mathcal{A}),K_\mathcal{A})$ is in its canonical form $(B({G}_\mathcal{A})_{\rm can},1)$. Let $(A,z)\in G_{\operatorname{degn}}$ such that
$(B({G}_\mathcal{B}),K_\mathcal{B})=(A,z)\cdot (B({G}_\mathcal{B})_{\rm can},1)$, then the continuity of the action of $G_{\operatorname{degn}}$ implies that $\lim\limits_{\epsilon \longrightarrow 0} \left((A,z)^{-1}(A_{\epsilon},z_{\epsilon})\right)\cdot (B({G}_\mathcal{A})_{\rm can},1)=(B({G}_\mathcal{B})_{\rm can},1)$.
\end{proof}

\begin{Proposition}\label{p2}
Let $\mathcal{A}$, $\mathcal{B}$ be free abstract quadratic algebras and $\epsilon \longmapsto (A_{\epsilon},1)$
a continuous map from $\C^*$ into $G_{\operatorname{degn}}$ such that
\begin{gather*}
 \lim_{\epsilon \longrightarrow 0}\widehat{A_{\epsilon}}^{\rm t} B({G}_\mathcal{A})_{\rm can} \widehat{A_{\epsilon}}=B({G}_\mathcal{B})_{\rm can}.
\end{gather*}
Then there is another continuous map $\epsilon \longmapsto ({C}_{\epsilon},z_{\epsilon})$ from $\C^*$ into $G_{\operatorname{degn}}$ such that $\lim\limits_{\epsilon \longrightarrow 0}(C_{\epsilon},z_{\epsilon})\cdot (B({G}_\mathcal{A})_{\rm can},1)=(B({G}_\mathcal{B})_{\rm can},1)$.
\end{Proposition}

\begin{proof} Def\/ine \begin{gather*}z_{\epsilon}:=\left((A_{\epsilon})_{1,1}(A_{\epsilon})_{2,2}-(A_{\epsilon})_{1,2}
(A_{\epsilon})_{2,1}\right)^{4}(A_{\epsilon})_{5,5}^{4},\\
C_{\epsilon}=\left((A_{\epsilon})_{1,1}(A_{\epsilon})_{2,2}-(A_{\epsilon})_{1,2}(A_{\epsilon})_{2,1}\right)^{-1}(A_{\epsilon})_{5,5}^
 {-1}A_{\epsilon}.\tag*{\qed}
 \end{gather*}\renewcommand{\qed}{}
\end{proof}

Note that Propositions \ref{p1} and \ref{p2} also hold for algebraic contractions. The conclusion from the last two proposition is that for the purpose of classifying contractions of free abstract quadratic algebras it is enough to take the Casimirs $B(G)_{\mathcal{A}}$ and $B(G)_{\mathcal{B}}$ in their canonical forms and to consider only the action of the group of invertible matrices of the form \begin{gather}
 \widehat{A}=\left(\begin{matrix}
A_{1,1} &A_{1,2} &A_{1,3} & A_{1,4} \\
A_{2,1} &A_{2,2} &A_{2,3} & A_{2,4} \\
0 &0 &A_{3,3} &0\\
0 &0 &0 & A_{4,4}
\end{matrix}\right),\label{fo72}
\end{gather} on symmetric matrices $B$ by $B\longmapsto \widehat{A}^{\rm t}B\widehat{A}$. We shall denote the space of 4 by 4
complex symmetric matrices by $\operatorname{Sym}(4,\C)$ and the group of matrices of the form (\ref{fo72}) by $\widehat{G}_{\operatorname{degn}}$. The group $\widehat{G}_{\operatorname{degn}}$ is a complex algebraic group and the space ${\rm Sym}(4,\C)$ is a complex algebraic variety on which $\widehat{G}_{\operatorname{degn}}$ acts algebraically. As was explained in \cite[Section~7.1.2]{EKMS2017}
this implies that if~$\mathcal{B}$ is a contraction of a quadratic algebra $\mathcal{A}$ then $\mathcal{A}$ is not a contraction of $\mathcal{B}$ (unless $\mathcal{A}$ and $\mathcal{B}$ are isomorphic). In addition if~$\mathcal{B}$ is a contraction of~$\mathcal{A}$ then the rank of any matrix that represents~$B(G)_{\mathcal{B}}$ and the rank of its upper left~2 by~2 block can not exceed the corresponding ranks for any matrix that represents~$B(G)_{\mathcal{A}}$. Hence there is certain hierarchy for contractions that is governed by the rank. By a rank of a free abstract quadratic algebra we mean $(\operatorname{rank}(B),\operatorname{rank}(b))$ of the corresponding matrices $B$ and $b$ in any bases. We shall make this more precise below.

\subsection[Classif\/ication of abstract algebraic contractions of superintegrable systems]{Classif\/ication of abstract algebraic contractions\\ of superintegrable systems}

Organized according to their rank, the canonical forms of the free abstract quadratic algebras of free triplets of 2D constant curvature
spaces and Darboux spaces are given in Table~\ref{tabl5} below.
\begin{table}[h]\centering
 \begin{tabular}{|l|l |l | l | p{9.2cm}|}
 \hline \multicolumn{5}{|c|}{canonical forms for the Casimirs of free degenerate 2D 2nd order }\\
 \multicolumn{5}{|c|}{superintegrable systems on constant curvature and Darboux spaces}\\
 \hline $\#$ & system & $\operatorname{rank}(B)$& $\operatorname{rank}(b)$ & canonical form
\\ \hline
1& $S6$& 4& 2 & $B^{22}(1,1)$, $\mathcal{L}_1^2+ \mathcal{L}_2^2+2\mathcal{H}{\cal X}^2+{\cal X}^4$\tsep{2pt}\bsep{2pt} \\ \hline
2&$E18$& 4& 2 & $B^{22}(0,1)$, $\mathcal{L}_1^2+ \mathcal{L}_2^2+2\mathcal{H}{\cal X}^2$\tsep{2pt}\bsep{2pt}\\ \hline
3&$D3E$& 4& 2& $B^{21}(1,0)$, $\mathcal{L}_1^2+ \mathcal{L}_2^2+{\mathcal{H}}^2+2\mathcal{H}{\cal X}^2$
\tsep{2pt}\bsep{2pt}\\ \hline
4&$D4(b)D$& 4& 2 & $B^{21}(1,-2)$, $\mathcal{L}_1^2+ \mathcal{L}_2^2+\mathcal{H}^2+2\mathcal{H}{\cal X}^2-2{\cal X}^4$\\ \hline
5&$S3$& 4& 2& $B^{21}(\sqrt{2}e^{i\frac{3\pi}{4}},-2i)$,\tsep{2pt}\bsep{2pt} $\mathcal{L}_1^2+ \mathcal{L}_2^2+\mathcal{H}^2+2\sqrt{2}e^{i\frac{3\pi}{4}}\mathcal{H}{\cal X}^2-2i{\cal X}^4$\\ \hline
6&$E3$ & 3& 2 & $B^{21}(0,0)$, $\mathcal{L}_1^2+ \mathcal{L}_2^2+\mathcal{H}^2$
\tsep{2pt}\bsep{2pt}\\ \hline
7&$E12$ &4& 1 & $B^{17}(1)$, $\mathcal{L}_1^2+2\mathcal{L}_2{\cal X}^2+{\mathcal{H}}^2$
\tsep{2pt}\bsep{2pt}\\ \hline
8&$D1D$ & 4& 1 & $B^{16}(1)$, $\mathcal{L}_1^2+2\mathcal{L}_2\mathcal{H}+\mathcal{H}^2$
\tsep{2pt}\bsep{2pt}\\ \hline
9&$D2D$ & 4& 1 & $B^{15}(1)$, $\mathcal{L}_1^2+2\mathcal{L}_2\mathcal{H}+2\mathcal{L}_2\mathcal{X}^2+2\mathcal{H}\mathcal{X}^2$
\tsep{2pt}\bsep{2pt}\\ \hline
10&$E6$ & 3& 1 & $B^{15}(0)$, $\mathcal{L}_1^2+2\mathcal{L}_2{\mathcal{H}}+2\mathcal{L}_2{\cal X}^2$\tsep{2pt}\bsep{2pt}\\ \hline
11&$E5$ & 3& 1 & $B^{11}(0,0,1)$, $\mathcal{L}_1^2+2\mathcal{H}{\cal X}^2+{\cal X}^4$ \tsep{2pt}\bsep{2pt}\\ \hline
12&$E14$ & 3& 1 & $B^{17}(0)$, $\mathcal{L}_1^2+2\mathcal{L}_2{\cal X}^2$
\tsep{2pt}\bsep{2pt}\\ \hline
13&$S5$ & 3& 1 & $B^{17}(0)$, $\mathcal{L}_1^2+2\mathcal{L}_2{\cal X}^2$
\tsep{2pt}\bsep{2pt}\\ \hline
14&$E13$ & 4& 0 & $B^{08}$, $2\mathcal{L}_1\mathcal{H}+2\mathcal{L}_2\mathcal{X}^2$
\tsep{2pt}\bsep{2pt}\\ \hline
15&$E4$ & 3& 0 & $B^{07}(1)$, $\mathcal{H}^2+2\mathcal{L}_2\mathcal{X}^2$
\tsep{2pt}\bsep{2pt}\\ \hline
\end{tabular}
\caption{List of canonical forms of free abstract degenerate quadratic algebras.}\label{tabl5}
\end{table}
We shall classify all possible algebraic contractions between any two free abstract quadratic algebras that appear in Table~\ref{tabl5}. More precisely for any two such algebras we determine if such a contraction is possible or not. If it does we give one realization of it, unless it is a~contraction from a free abstract quadratic algebra to itself. We shall start with some general observations. We note that the quadratic algebras of ${E}14$ and $S5$ coincide so we only keep $E14$ in our notation. We divide the free abstract quadratic algebras of the second order free degenerate superintegrable systems on 2D constant curvature spaces and 2D Darboux spaces according to their ranks: $R_{4,2}:=\{{S}6,{E}18,{D}3E,{D}4(b)D,{S}3\}$,
 $R_{3,2}:=\{{E}3\}$,
 $R_{4,1}:=\{{E}12,{D}1D,{D}2D\}$,
 $R_{3,1}:=\{{E}6,{E}5,{E}14\}$,
 $R_{4,0}:=\{{E}13\}$,
 $R_{3,0}:=\{{E}4\}$.
By the discussion above, besides contractions between two algebras in the same class $R_{i,j}$ we can potentially f\/ind a contraction only according to the following diagram:
\begin{gather*}\xymatrix{
 &R_{4,2}\ar[dr]\ar[dl]&\\
 R_{3,2}\ar[rd]&&R_{4,1}\ar[d]\ar[ld]\\
&R_{3,1}\ar[d]&R_{4,0}\ar[ld] \\
&R_{3,0}.& } \end{gather*}
The classif\/ication is summarized in Table~\ref{tab2}. Detailed analysis is presented below. Representatives for all possible contractions are given in Section~\ref{exp}.
 \begin{table}\centering\footnotesize
\begin{tabular}{|l|l|l|l|l|l|l|l|l|l|l|l|l|l|l|}
\hline & ${S}6$& ${E}18$& ${D}3E$ & ${D}4(b)D$ & ${S}3$ & ${E}3$ & ${E}12$& ${D}1D$& ${D}2D$ & ${E}6$ & ${E}5$& ${E}14$& ${E}13$ & ${E}4$
\\ \hline
 ${S}6$& + &+ & \textbf{--} & \textbf{--} & \textbf{--}& \textbf{--}&+ & \textbf{--} & \textbf{--}& \textbf{--}&+ & + & +&+
\\ \hline
 ${E}18$& \textbf{--} & + & \textbf{--} & \textbf{--} & \textbf{--}& \textbf{--}& \textbf{--} & \textbf{--} & \textbf{--} & \textbf{--} &+ & \textbf{--} & +&+
\\ \hline
${D}3E$
& \textbf{--} &+ & + & \textbf{--} & \textbf{--}&+& \textbf{--} & + & \textbf{--} & \textbf{--} &+ & \textbf{--} &+ &+
\\ \hline
${D}4(b)D$
& \textbf{--} & \textbf{--} & \textbf{--} & + & \textbf{--}&+& + &+ & \textbf{--} & \textbf{--} &+ &+ &+ &+
\\ \hline
${S}3$& \textbf{--} & \textbf{--} & \textbf{--} & \textbf{--} &+&+& \textbf{--} & + &\textbf{--} &+ &+ & + & +&+
\\ \hline
${E}3$& \textbf{--} & \textbf{--} & \textbf{--} & \textbf{--} &\textbf{--} &+& \textbf{--} & \textbf{--} &\textbf{--} & \textbf{--} &+ & \textbf{--} &\textbf{--} &+
\\ \hline
${E}12$& \textbf{--} & \textbf{--} & \textbf{--} & \textbf{--} & \textbf{--} & \textbf{--} & +& \textbf{--} & \textbf{--}& \textbf{--} &+ &+ & +&+
\\ \hline
${D}1D$
& \textbf{--} & \textbf{--} & \textbf{--} & \textbf{--} & \textbf{--} & \textbf{--} & \textbf{--} & + & \textbf{--}& \textbf{--} &+ & \textbf{--} &+ &+
\\ \hline
${D}2D$
& \textbf{--} & \textbf{--} & \textbf{--} & \textbf{--} & \textbf{--} & \textbf{--} & + & + & +& + & +& + & +&+
\\ \hline
 ${E}6$& \textbf{--} & \textbf{--} & \textbf{--} & \textbf{--} & \textbf{--} & \textbf{--} & \textbf{--}& \textbf{--} &\textbf{--} & +&+ & + &\textbf{--} &+
\\ \hline
 ${E}5$& \textbf{--} & \textbf{--} & \textbf{--} & \textbf{--} & \textbf{--} & \textbf{--} & \textbf{--} & \textbf{--} & \textbf{--} & \textbf{--} & +& \textbf{--} & \textbf{--} &+
\\ \hline ${E}14$& \textbf{--} & \textbf{--} & \textbf{--} & \textbf{--} & \textbf{--} & \textbf{--} & \textbf{--} & \textbf{--} & \textbf{--} & \textbf{--} &+ & + & \textbf{--} & +\\ \hline
 ${E}13$& \textbf{--} & \textbf{--} & \textbf{--} & \textbf{--} & \textbf{--} & \textbf{--} & \textbf{--} & \textbf{--} & \textbf{--} & \textbf{--} & \textbf{--} & \textbf{--} &+ &+
\\ \hline
 ${E}4$& \textbf{--} & \textbf{--} & \textbf{--} & \textbf{--} & \textbf{--} & \textbf{--} & \textbf{--} & \textbf{--} & \textbf{--} & \textbf{--} & \textbf{--} & \textbf{--} & \textbf{--} &+
\\ \hline\end{tabular}
\caption{List of algebraic contractions between free abstract quadratic algebras of 2D free degenerate superintegrable systems on constant curvature spaces and Darboux spaces. A plus in the rubric placed in the $i$-th row and $j$-th column indicates that there is a contraction from the system listed in the $i$-th row of the f\/irst column to the system listed in the $j$-th column of the f\/irst row. A minus indicates that there is no such contraction.}\label{tab2}
\end{table}

\subsubsection{Contractions between rank two free abstract quadratic algebras}

Note that for any of the $B_{\rm can} $ matrices for the systems
 $R_{4,2}\cup R_{3,2}=\{{S}6,{E}18,{D}3E,{D}4(b)D,{S}3\}$ $\cup \{{E}3\}$ we have
\begin{gather}\label{ee68}
 \widehat{A}^{\rm t}B \widehat{A}=\left(\begin{matrix}
r^{\rm t}& 0\\
s^{\rm t}&t
\end{matrix}\right) \left(\begin{matrix}
1&0 \\
0 &d
\end{matrix}\right)\left(\begin{matrix}
r&s \\
 0&t
\end{matrix}\right)=\left(\begin{matrix}
r^{\rm t}r&r^{\rm t}s \\
s^{\rm t}r&s^{\rm t}s+tdt
\end{matrix}\right).
\end{gather}
 Consider $A=A(\epsilon)$ in equation (\ref{ee68}) such that
\begin{gather}\label{114}
\lim_{\epsilon \longrightarrow 0}\widehat{A(\epsilon)}^{\rm t}B \widehat{A(\epsilon)}=B^{0},
\end{gather}
where $B^0\in R_{4,2}\cup R_{3,2}$. Below we prove that without loss of generality we can assume that $r^{\rm t}_{\epsilon}r_{\epsilon}$ is diagonal and $s_{\epsilon}$ is the zero two by two matrix.

\begin{Proposition}\label{pro5.3}
Suppose that $ \{B_{\epsilon} \}_{\epsilon\in \mathbb{R}^+}$ is a family of $4\times 4$ symmetric matrices with entries in $\C[\epsilon]$ and such that
\begin{gather*}\lim_{\epsilon \longrightarrow 0^+}\left(\begin{matrix}
({B}_{\epsilon})_{11}&({B}_{\epsilon})_{12}& ({B}_{\epsilon})_{13}& ({B}_{\epsilon})_{14}\\
({B}_{\epsilon})_{12}&({B}_{\epsilon})_{22}& ({B}_{\epsilon})_{23}& ({B}_{\epsilon})_{24}\\
({B}_{\epsilon})_{13}& ({B}_{\epsilon})_{23}&({B}_{\epsilon})_{33}& ({B}_{\epsilon})_{34}\\
({B}_{\epsilon})_{14}&({B}_{\epsilon})_{24}&({B}_{\epsilon})_{34}& ({B}_{\epsilon})_{44}
 \end{matrix}\right)=\left(\begin{matrix}
1&0& 0& 0\\
0&1& 0& 0\\
0& 0&l_{33}& l_{34}\\
0&0&l_{34}& l_{44}
 \end{matrix}\right).
\end{gather*}
Then there exists a continuous function $\epsilon \mapsto A_{\epsilon} $ from $\C^*$ to $\widehat{G}_{\operatorname{degn}}$ with each entry polynomial in $\epsilon$ such that $\widehat{A}_{\epsilon}^{\rm t}B_{\epsilon} \widehat{A}_{\epsilon}$ is of the form \begin{gather*}\left(\begin{matrix}
(\widetilde{B}_{\epsilon})_{11}&0& 0& 0\\
0&(\widetilde{B}_{\epsilon})_{22}& 0& 0\\
0& 0&(\widetilde{B}_{\epsilon})_{33}& (\widetilde{B}_{\epsilon})_{34}\\
0&0 &(\widetilde{B}_{\epsilon})_{34}& (\widetilde{B}_{\epsilon})_{44}
 \end{matrix}\right),\end{gather*}
and
\begin{gather*}
 \lim_{\epsilon \longrightarrow 0}\widehat{A}_{\epsilon}^{\rm t}B_{\epsilon} \widehat{A}_{\epsilon}=\lim_{\epsilon \longrightarrow 0}B_{\epsilon}.
\end{gather*}
\end{Proposition}

\begin{proof}
Note that for $\widehat{A}_{\epsilon}^{\rm t}=\left(\begin{smallmatrix}
1&0& 0& 0\\
-({B}_{\epsilon})_{12}&({B}_{\epsilon})_{11}& 0& 0\\
-({B}_{\epsilon})_{13}& 0&({B}_{\epsilon})_{11}& 0\\
-({B}_{\epsilon})_{14}&0&0& ({B}_{\epsilon})_{11}
 \end{smallmatrix}\right)$
 $\lim\limits_{\epsilon \longrightarrow 0}\widehat{A}_{\epsilon}^{\rm t}$ is the identity matrix and $\widehat{A}_{\epsilon}^{\rm t}B_{\epsilon} \widehat{A}_{\epsilon}$ is a symmetric matrix with all entries in the f\/irst row and f\/irst column besides the $(1,1)$ entry equal to zero. Similar matrix take care of the second row and second column.
 \end{proof}

Now suppose that equation (\ref{114}) holds. Using the last proposition this means that
we can assume that $ \widehat{A(\epsilon)}=\left(\begin{smallmatrix}
r_{\epsilon}&0 \\
 0&t_{\epsilon}
\end{smallmatrix}\right)$ with $r_{\epsilon}\in {\rm GL}(2,\C)$ and, as always, $t_{\epsilon}$ a diagonal invertible matrix.
From this we easily see that the only possible contractions between two algebras in $R_{4,2}\cup R_{3,2}$ are $S6\longrightarrow E18$, $D3E\longrightarrow E18$, $D3E\longrightarrow E3$, $D4(b)D\longrightarrow E3$, $S3\longrightarrow E3$.

\subsection[Contractions of rank two free abstract quadratic algebras to rank one algebras]{Contractions of rank two free abstract quadratic algebras\\ to rank one algebras}
\begin{Proposition}
Suppose that $\{B_{\epsilon}\}_{\epsilon\in \mathbb{R}^+}$ is a family of $4\times 4$ symmetric matrices with entries in
$\C[\epsilon]$ such that
\begin{gather*}
\lim_{\epsilon \longrightarrow 0^+}
 B_{\epsilon}=L=\left(\begin{matrix}
1& 0& l_{13}& l_{14}\\
0&0& l_{23}& l_{24}\\
 l_{13}& l_{23}&l_{33}& l_{34}\\
 l_{14}&l_{24}&l_{34}& l_{44}
 \end{matrix}\right).
\end{gather*}
Then there exists a continuous function $\epsilon \mapsto A_{\epsilon} $ from $\C^*$ to $\widehat{G}_{\operatorname{degn}}$ with each entry polynomial in~$\epsilon$ such that $\widehat{A}_{\epsilon}^{\rm t}B_{\epsilon} \widehat{A}_{\epsilon}$ is of the form \begin{gather*}\left(\begin{matrix}
(\widetilde{B}_{\epsilon})_{11}&0& 0& 0\\
0&(\widetilde{B}_{\epsilon})_{22}& (\widetilde{B}_{\epsilon})_{23}& (\widetilde{B}_{\epsilon})_{24}\\
0& (\widetilde{B}_{\epsilon})_{23}&(\widetilde{B}_{\epsilon})_{33}& (\widetilde{B}_{\epsilon})_{34}\\
0&(\widetilde{B}_{\epsilon})_{24} &(\widetilde{B}_{\epsilon})_{34}& (\widetilde{B}_{\epsilon})_{44}
 \end{matrix}\right),\end{gather*}
and
\begin{gather*}
 \lim_{\epsilon \longrightarrow 0}\widehat{A}_{\epsilon}^{\rm t}B_{\epsilon} \widehat{A}_{\epsilon}=\lim_{\epsilon \longrightarrow 0}B_{\epsilon}.
\end{gather*}
If in addition $l_{13}=0$ then $(\widetilde{B}_{\epsilon})_{13}=0$ and similarly if $l_{14}=0$ then $(\widetilde{B}_{\epsilon})_{14}=0$.
 \end{Proposition}
The proof is similar to the proof of Proposition~\ref{pro5.3}.

\begin{Corollary}
In any contraction from a rank two systems $R_{4,2}\cup R_{3,2}=\{{S}6,{E}18,{D}3E,$ ${D}4(b)D,{S}3,{E}3\}$ to of one of the rank one systems $R_{4,1}\cup R_{3,1}=\{{E}12,{D}1D,{D}2D,{E}6,{E}5,{E}14\}$ we can assume that
\begin{gather}
\widehat{A(\epsilon)}^{\rm t}B \widehat{A(\epsilon)}=
\left(\begin{matrix}
({B}_{\epsilon})_{11}&0&0& 0\\
0&({B}_{\epsilon})_{22}& ({B}_{\epsilon})_{23}& ({B}_{\epsilon})_{24}\\
0& ({B}_{\epsilon})_{23}&({B}_{\epsilon})_{33}& ({B}_{\epsilon})_{34}\\
0&({B}_{\epsilon})_{24}&({B}_{\epsilon})_{34}& ({B}_{\epsilon})_{44}
 \end{matrix}\right). \label{77n}
\end{gather}
\end{Corollary}

\subsubsection[Contractions of $D3E$ to rank one algebras]{Contractions of $\boldsymbol{D3E}$ to rank one algebras}
For $A(\epsilon)\in \widehat{G}_{\operatorname{degn}}$ the matrix $\widehat{A(\epsilon)}^{\rm t} B( {G} )^{21}(1, 0) \widehat{A(\epsilon)}$ is given by
\begin{gather*}
 \left(\begin{smallmatrix}
A_{11}^2 + A_{21}^2 &A_{11}A_{12}  + A_{21}A_{22} &A_{11}A_{13}  + A_{21}A_{23} &A_{11}A_{14}  + A_{21}A_{24} \\
A_{12}A_{11}  + A_{22}A_{21} &A_{12} ^2 + A_{22}^2 &A_{12}A_{13}  + A_{22}A_{23} & A_{12}A_{14}  + A_{22}A_{24} \\
A_{11}A_{13}  + A_{21}A_{23} &A_{12}A_{13}  + A_{22}A_{23} &A_{13}^2 + A_{23}^2 + A_{33}^2 &A_{13}A_{14}  + A_{23}A_{24} + A_{33}A_{44}\\
A_{11}A_{14}  + A_{21}A_{24} &A_{12}A_{14}  + A_{22}A_{24} &A_{13}A_{14}  + A_{23}A_{24} + A_{33}A_{44} & A_{14} ^2 + A_{24}^2
\end{smallmatrix}\right).
\end{gather*}
Assuming that this matrix is in the form of (\ref{77n}) then this implies that exist $\beta$, $\gamma$, $\delta$ complex valued functions of $\epsilon$ def\/ine on $\C^*$ such that
 \begin{gather*}
(A_{12},A_{22})=\beta (A_{21},-A_{11}),\! \qquad  (A_{13},A_{23})=\gamma (A_{21},-A_{11}), \!\qquad
(A_{14},A_{24})=\delta (A_{21},-A_{11}).
\end{gather*}
Hence $\widehat{A(\epsilon)}^{\rm t} B( {G} )^{21}(1, 0) \widehat{A(\epsilon)}$ takes the form
 \begin{gather*}
 \left(\begin{matrix}
A_{11}^2+A_{21}^2 &0&0&0 \\
0 &\beta^2(A_{11}^2+A_{21}^2) &\beta\gamma(A_{11}^2+A_{21}^2) & \beta\delta(A_{11}^2+A_{21}^2) \\
0 &\beta\gamma(A_{11}^2+A_{21}^2) & \gamma^2(A_{11}^2+A_{21}^2)+A_{33}^2 & \gamma\delta(A_{11}^2+A_{21}^2)+A_{33}A_{44}\\
0 & \beta\delta(A_{11}^2+A_{21}^2)& \gamma\delta(A_{11}^2+A_{21}^2)+A_{33}A_{44} & \delta^2(A_{11}^2+A_{21}^2)
\end{matrix}\right).
\end{gather*}
We are going to consider the limit of $\epsilon$ goes to zero of such a matrix and check if it can be equal to one of the
matrices of the canonical forms of the systems in $R_{4,1}\cup R_{3,1}$. This means that
$\lim\limits_{\epsilon \longrightarrow 0^+} (A_{11}^2+A_{21}^2)= 1$ and we can replace $\widehat{A(\epsilon)}^{\rm t} B( {G} )^{21}(1, 0)
\widehat{A(\epsilon)}$ by
\begin{gather*}
 \left(\begin{matrix}
1 &0&0&0 \\
0 &\beta^2 &\beta\gamma & \beta\delta \\
0 &\beta\gamma & \gamma^2+A_{33}^2 & \gamma\delta+A_{33}A_{44}\\
0 & \beta\delta& \gamma\delta+A_{33}A_{44} & \delta^2
\end{matrix}\right),
\end{gather*}
 with $\lim\limits_{\epsilon \longrightarrow 0^+} \beta= 0 $. From this we can show that $D3E$ can not be contracted to $E6$, $E12$, $E14$, $D2D$.

\subsubsection[Contractions of $E18$ to rank one algebras]{Contractions of $\boldsymbol{E18}$ to rank one algebras}
Following the same steps as in the case of $D3E$ we can replace $\widehat{A(\epsilon)}^{\rm t} B( {G} )^{22}(0, 1) \widehat{A(\epsilon)}$ by
\begin{gather*}
\left(\begin{matrix}
1 &0&0&0 \\
0 &\beta^2 &\beta\gamma & \beta\delta \\
0 &\beta\gamma & \gamma^2 & \gamma\delta+A_{33}A_{44}\\
0 & \beta\delta& \gamma\delta+A_{33}A_{44} & \delta^2
\end{matrix}\right),
\end{gather*}
 with $\lim\limits_{\epsilon \longrightarrow 0^+} \beta= 0 $. From this we can show that $E18$ can not be contracted to $E6$, $E12$, $D1D$, $D2D$.

\subsubsection[Contractions of $S6$ to rank one algebras]{Contractions of $\boldsymbol{S6}$ to rank one algebras}
Following the same steps as in the case of $D3E$
we can replace $\widehat{A(\epsilon)}^{\rm t} B( {G} )^{22}(1, 1) \widehat{A(\epsilon)}$ by
\begin{gather*}
 \left(\begin{matrix}
1 &0&0&0 \\
0 &\beta^2 &\beta\gamma & \beta\delta \\
0 &\beta\gamma & \gamma^2 & \gamma\delta+A_{33}A_{44}\\
0 & \beta\delta& \gamma\delta+A_{33}A_{44} & \delta^2+A_{44}^2
\end{matrix}\right),
\end{gather*}
 with $\lim\limits_{\epsilon \longrightarrow 0^+} \beta= 0 $. From this we can show that $S6$ can not be contracted to $E6$, $D1D$, $D2D$.

 \subsubsection[Contractions of $D4(b)D$ to rank one algebras]{Contractions of $\boldsymbol{D4(b)D}$ to rank one algebras}
Following the same steps as in the case of $D3E$ we can replace $\widehat{A(\epsilon)}^{\rm t} B( {G} )^{21}(1,-2)\widehat{A(\epsilon)}$ by
\begin{gather*}
 \left(\begin{matrix}
1 &0&0&0 \\
0 &\beta^2 &\beta\gamma & \beta\delta \\
0 &\beta\gamma & \gamma^2+A_{33}^2 & \gamma\delta+A_{33}A_{44}\\
0 & \beta\delta& \gamma\delta+A_{33}A_{44} & \delta^2-2A_{44}^2
\end{matrix}\right),
\end{gather*}
 with $\lim\limits_{\epsilon \longrightarrow 0^+} \beta= 0 $. From this we can show that $D4(b)D$ can not be contracted to $D2D$, $E6$.

 \subsubsection[Contractions of $S3$ to rank one algebras]{Contractions of $\boldsymbol{S3}$ to rank one algebras}
Following the same steps as in the case of $D3E$
we can replace $\widehat{A(\epsilon)}^{\rm t} B( {G} )^{21}(\sqrt{2}e^{i\frac{3\pi}{4}},-2i)\widehat{A(\epsilon)}$ by
\begin{gather*}
 \left(\begin{matrix}
1 &0&0&0 \\
0 &\beta^2 &\beta\gamma & \beta\delta \\
0 &\beta\gamma & \gamma^2+A_{33}^2 & \gamma\delta+\sqrt{2}e^{i\frac{3\pi}{4}}A_{33}A_{44}\\
0 & \beta\delta& \gamma\delta+\sqrt{2}e^{i\frac{3\pi}{4}}A_{33}A_{44} & \delta^2-2iA_{44}^2
\end{matrix}\right),
\end{gather*}
 with $\lim\limits_{\epsilon \longrightarrow 0^+} \beta= 0 $. From this we can show that $S3$ can not be contracted to $D2D$, $E12$.

 \subsubsection[Contractions of $E3$ to rank one algebras]{Contractions of $\boldsymbol{E3}$ to rank one algebras}
Following the same steps as in the case of $D3E$ we can replace $\widehat{A(\epsilon)}^{\rm t} B( {G} )^{21}(0,0)\widehat{A(\epsilon)}$ by
\begin{gather*}
 \left(\begin{matrix}
1 &0&0&0 \\
0 &\beta^2 &\beta\gamma & \beta\delta \\
0 &\beta\gamma & \gamma^2 & \gamma\delta \\
0 & \beta\delta& \gamma\delta & \delta^2
\end{matrix}\right),
\end{gather*}
 with $\lim\limits_{\epsilon \longrightarrow 0^+} \beta= 0 $. From this we can show that $E3$ can not be contracted to $E6$ and $E14$.

\subsubsection{Explicit contractions}\label{exp}
\noindent\textbf{E13:}
\begin{gather*}\begin{array}{@{}c@{}}
 {E}13 \longrightarrow {E}4\\
\left(\begin{matrix}
\epsilon &0 &\frac{1}{2} & 0 \\
0 &1 &0& 0 \\
0 &0 &1 &0\\
0 &0 &0 &1
\end{matrix}\right).\end{array}
\end{gather*}
\noindent\textbf{E14:}
\begin{gather*}\begin{array}{@{}c@{}c@{}c@{}}
 {E}14 \longrightarrow {E}4 &&  {E}14 \longrightarrow {E}5\\
 \left(\begin{matrix}
\epsilon &0 &1 & 0 \\
0 &1 &0& 0 \\
0 &0 &1 &0\\
0 &0 &0 &1
\end{matrix}\right), & \qquad &  \left(\begin{matrix}
1 &0 &0 & 0 \\
0 &\epsilon &1& \frac{1}{2} \\
0 &0 &1 &0\\
0 &0 &0 &1
\end{matrix}\right).\end{array}
\end{gather*}
\textbf{E5:}
\begin{gather*}\begin{array}{@{}c@{}}
 {E}5\longrightarrow {E}4\\
  \left(\begin{matrix}
\epsilon^2 &\epsilon &1 & \epsilon^{-1} \\
0 &1 &0& 0 \\
0 &0 &i &0\\
0 &0 &0 &i\epsilon^{-1}
\end{matrix}\right).\end{array}
\end{gather*}
\textbf{{D}2D:}
\begin{gather*}\begin{array}{@{}c@{}c@{}c@{}c@{}c@{}}
{D}2D\longrightarrow {E}6 && {D}2D\longrightarrow {E}12 && {D}2D\longrightarrow {D}1D \\
  \left(\begin{matrix}
1 &0&1 & 0\\
0 &\epsilon^{-1} &0& 0 \\
0 &0 &\epsilon &0\\
0 &0 &0 &\epsilon
\end{matrix}\right), &\qquad & \left(\begin{matrix}
1 &0&0 & 0\\
0 &\epsilon&-\frac{i}{\sqrt{2}}& 0 \\
0 &0 &\frac{i}{\sqrt{2}} &0\\
0 &0 &0 &\epsilon^{-1}
\end{matrix}\right),&\qquad & \left(\begin{matrix}
1 &0&0 & 0\\
0 &\epsilon&\frac{\epsilon}{2}& 0 \\
0 &0 &\epsilon^{-1} &0\\
0 &0 &0 &\epsilon^{2}
\end{matrix}\right),\end{array}\\
\begin{array}{@{}c@{}c@{}c@{}c@{}c@{}c@{}c@{}}
{D}2D\longrightarrow {E}5 && {D}2D\longrightarrow {E}14 && {D}2D\longrightarrow {E}4 && {D}2D\longrightarrow {E}13\\
\left(\begin{matrix}
1 &0&0 & 0\\
0 &\epsilon&0& \frac{1}{2\epsilon} \\
0 &0 &2\epsilon &0\\
0 &0 &0 &\epsilon
\end{matrix}\right),&\qquad & \left(\begin{matrix}
1 &0&0 & 0\\
0 &1&0&0 \\
0 &0 &1&0\\
0 &0 &0 &\epsilon
\end{matrix}\right), &\qquad & \left(\begin{matrix}
\epsilon &0&1 & 1\\
0 &\epsilon&0&0 \\
0 &0 &\epsilon&0\\
0 &0 &0 &\epsilon
\end{matrix}\right),&\qquad & \left(\begin{matrix}
0 &\epsilon&0 & \epsilon^{-1}\\
1 &0&0&0 \\
0 &0 &1&0\\
0 &0 &0 &i\epsilon^{-1}
\end{matrix}\right).\end{array}
\end{gather*}
\textbf{{E}6:}
\begin{gather*}\begin{array}{@{}c@{}c@{}c@{}c@{}c@{}}
 {E}6 \longrightarrow {E}14  && {E}6 \longrightarrow {E}5 &&  {E}6 \longrightarrow {E}4\\
 \left(\begin{matrix}
1 &0 &0 & 0 \\
0 &1 &0& 0 \\
0 &0 &\epsilon &0\\
0 &0 &0 & 1
\end{matrix}\right),&\qquad & \left(\begin{matrix}
1 &0 &0 & 0 \\
0 &\epsilon &\epsilon& \epsilon^{-1} \\
0 &0 &\epsilon &0\\
0 &0 &0 & \frac{1}{2}\epsilon
\end{matrix}\right),&\qquad & \left(\begin{matrix}
\epsilon &0 &1 & 0 \\
0 &1 &0& 0 \\
0 &0 &\epsilon &0\\
0 &0 &0 & 1
\end{matrix}\right).\end{array}
\end{gather*}
\textbf{E12:}
\begin{gather*}\begin{array}{@{}c@{}c@{}c@{}c@{}c@{}c@{}c@{}}
 {E}12 \longrightarrow {E}14 && {E}12 \longrightarrow {E}5 &&  {E}12 \longrightarrow {E}4 && {E}12 \longrightarrow {E}13 \\
 \left(\begin{matrix}
1 &0 &0 & 0 \\
0 &1 &0& 0 \\
0 &0 &\epsilon &0\\
0 &0 &0 &1
\end{matrix}\right),&\qquad &\left(\begin{matrix}
1 &0 &0 & 1 \\
-1 &\epsilon &1& 0 \\
0 &0 &\epsilon &0\\
0 &0 &0 &1
\end{matrix}\right),&\qquad &\left(\begin{matrix}
\epsilon &0 &0 & 0 \\
0 &1 &0& 0 \\
0 &0 &1 &0\\
0 &0 &0 &1
\end{matrix}\right),&\qquad &\left(\begin{matrix}
\epsilon &0 &\epsilon^{-1} & 0 \\
0 &1 &0& 0 \\
0 &0 &i\epsilon^{-1}&0\\
0 &0 &0 &1
\end{matrix}\right).\end{array}
\end{gather*}
The non-contraction of ${E}12$ to ${D}1D$, $E6$, ${D}2D$:
Demanding that $\lim\limits_{\epsilon \longrightarrow 0}\widehat{A(\epsilon)}^{\rm t} B^{17}(1)\widehat{A(\epsilon)}$ converge to the Canonical form of one of the systems ${D}1D$, $E6$, ${D}2D$, and observing the entries on places $(1,1)$, $(2,2)$, $(3,1)$, and $(3,2)$ in the matrix equation we obtain the equations:
\begin{gather*}
 \lim_{\epsilon \longrightarrow 0} A^2_{1,1} = 1, \qquad
 \lim_{\epsilon \longrightarrow 0} A_{1,2}^2=0,\qquad
 \lim_{\epsilon \longrightarrow 0} A_{1,3}A_{1,1}=0,\qquad
 \lim_{\epsilon \longrightarrow 0} A_{1,2}A_{1,3}=1,
\end{gather*} which obviously can not hold simultaneously.

\noindent
\textbf{D1D:}
\begin{gather*}\begin{array}{@{}c@{}c@{}c@{}c@{}c@{}}
 {D}1D \longrightarrow {E}5 & & {D}1D \longrightarrow {E}13 && {D}1D \longrightarrow {E}4 \\
 \left(\begin{matrix}
1 &0 &0& 0 \\
0 &\epsilon &0& 1 \\
0 &0 &1&0\\
0 &0 &0 &1
\end{matrix}\right),& \qquad & \left(\begin{matrix}
\epsilon^2 &\epsilon &0& \epsilon^{-1} \\
\epsilon^{-1} &\epsilon &0& 0 \\
0 &0 &\epsilon&0\\
0 &0 &0 &i\epsilon^{-1}
\end{matrix}\right), &\qquad & \left(\begin{matrix}
\epsilon^2 &\epsilon &0& \epsilon^{-1} \\
0&\epsilon &1/2& 0 \\
0 &0 &1&0\\
0 &0 &0 &i\epsilon^{-1}
\end{matrix}\right).\end{array}
\end{gather*}
 The non-contraction of ${D}1D $ to ${E}6$, ${E}12$, ${S}5$: Demanding that $\lim\limits_{\epsilon \longrightarrow 0}\widehat{A(\epsilon)}^{\rm t} B^{16}(1)\widehat{A(\epsilon)}$ converge to the canonical form of one of the systems ${E}6$, ${E}12$, ${S}5$, and observing the entries on places $(1,1)$, $(1,2)$, $(4,1)$, and $(4,2)$ in the matrix equation we obtain the equations:
\begin{gather*}
 \lim_{\epsilon \longrightarrow 0} A_{1,1} =\pm 1,\qquad
 \lim_{\epsilon \longrightarrow 0} A_{1,2}=0,\qquad
 \lim_{\epsilon \longrightarrow 0} A_{1,2}A_{1,4}=1,\qquad
 \lim_{\epsilon \longrightarrow 0} A_{1,1}A_{1,4}=0,
\end{gather*} which obviously can not hold simultaneously.

\noindent
\textbf{E3:}
\begin{gather*}\begin{array}{@{}c@{}c@{}c@{}}
 {E}3 \longrightarrow {E}5 && {E}3 \longrightarrow {E}4\\
 \left(\begin{matrix}
1 &0 &0 & 0 \\
0 &\epsilon &1& 1\\
0 &0 &i &0\\
0 &0 &0 &1
\end{matrix}\right),&\qquad & \left(\begin{matrix}
\epsilon &0 &1 & -i\epsilon^{-1} \\
i\epsilon &\epsilon &i &\epsilon^{-1}\\
0&0&1&0\\
0 &0 &0 &1
\end{matrix}\right).\end{array}
\end{gather*}
\textbf{D4(b)D:}
\begin{gather*}\begin{array}{@{}c@{}c@{}c@{}c@{}c@{}}
 {D}4(b)D \longrightarrow {E}4 && {D}4(b)D \longrightarrow {E}14 &&  {D}4(b)D \longrightarrow {E}5\\
 \left(\begin{matrix}
\epsilon &0 &1 & 0 \\
0 &\epsilon & 0 & \epsilon^{-1} \\
0 &0 &\epsilon^2 &0\\
0 &0 &0 &\frac{1}{\sqrt{2}\epsilon}
\end{matrix}\right),&\qquad & \left(\begin{matrix}
1 &0 &0 & 0 \\
0 &\epsilon & 0& \epsilon^{-1} \\
0 &0 &\epsilon^2 &0\\
0 &0 &0 &\frac{1}{\sqrt{2}\epsilon}
\end{matrix}\right),&\qquad & \left(\begin{matrix}
1 &0 &0 & 0 \\
0 &\epsilon &1& 1\\
0 &0 &i &0\\
0 &0 &0 &\epsilon
\end{matrix}\right),\end{array}\\
\begin{array}{@{}c@{}c@{}c@{}c@{}c@{}}
  {D}4(b)D \longrightarrow D1D &&
 {D}4(b)D \longrightarrow E12 && {D}4(b)D \longrightarrow E13  \\
\left(\begin{matrix}
1 &0 &0 & 0 \\
0 &-\epsilon &-\epsilon^{-1}& 0\\
0 &0 &i\epsilon^{-1} &0\\
0 &0 &0 &\epsilon^2
\end{matrix}\right), &\qquad & \left(\begin{matrix}
1 &0 &0 & 0 \\
0 &-\epsilon &\frac{1}{\sqrt{3}}& -\epsilon^{-1}\\
0 &0 &\sqrt{\frac{2}{3}} &0\\
0 &0 &0 &\frac{1}{\sqrt{2}\epsilon}
\end{matrix}\right),& \qquad & \left(\begin{matrix}
\epsilon &\epsilon &\frac{1}{2}\epsilon^{-1} & \frac{1}{2}\epsilon^{-1} \\
i\epsilon &-i\epsilon &-\frac{i}{2}\epsilon^{-1} & \frac{i}{2}\epsilon^{-1} \\
0 &0 &\epsilon&0\\
0 &0 &0 &\epsilon
\end{matrix}\right),\end{array}\\
\begin{array}{@{}c@{}}
 {D}4(b)D \longrightarrow E3 \\
\left(\begin{matrix}
1 &0 &0 & 0 \\
0 &1 & 0 &0 \\
0 &0 &1 &0\\
0 &0 &0 &\epsilon
\end{matrix}\right).\end{array}
\end{gather*}
\textbf{S3:}
\begin{gather*}\begin{array}{@{}c@{}c@{}c@{}c@{}c@{}}
 {S}3 \longrightarrow {E}5 && {S}3 \longrightarrow {E}14 && {S}3\longrightarrow {E}4 \\
 \left(\begin{matrix}
1 &0 &0 & 0 \\
0 &\epsilon & i & 0 \\
0 &0 &1 &0\\
0 &0 &0 &\frac{1}{\sqrt{2}}e^{-i\frac{3\pi}{4}}
\end{matrix}\right),&\qquad & \left(\begin{matrix}
1&0 &\epsilon & -\epsilon \\
0 &\epsilon & \epsilon^{3}& \epsilon^{-1} \\
0 &0 &\epsilon^2 &0\\
0 &0 &0 &\frac{i}{\epsilon\sqrt{2}}e^{-i\frac{3\pi}{4}}
\end{matrix}\right),&\qquad & \left(\begin{matrix}
\epsilon &0 &1 & 0 \\
0 &\epsilon & 0 & \epsilon^{-1} \\
0 &0 &\epsilon^2 &0\\
0 &0 &0 &i\frac{1}{\epsilon \sqrt{2}}e^{-i\pi/4}
\end{matrix}\right),\end{array}\\
\begin{array}{@{}c@{}c@{}c@{}c@{}c@{}} {S}3\longrightarrow D1D && S3 \longrightarrow E13 && S3 \longrightarrow E6\\ \nonumber
 \left(\begin{matrix}
1 &0 &0 & 0 \\
0 &-\epsilon & - \epsilon^{-1} &0 \\
0 &0 &i \epsilon^{-1} &0\\
0 &0 &0 & \epsilon^2
\end{matrix}\right),&\qquad & \left(\begin{matrix}
\epsilon &\epsilon &\frac{1}{2}\epsilon^{-1} & \frac{1}{2}\epsilon^{-1} \\
i\epsilon &-i\epsilon &-\frac{i}{2}\epsilon^{-1} & \frac{i}{2}\epsilon^{-1} \\
0 &0 &\epsilon&0\\
0 &0 &0 &\epsilon
\end{matrix}\right),& \qquad & \left(\begin{matrix}
1 &0 &0 &0\\
0 &-\epsilon &-\epsilon^{-1} &-\epsilon^{-1} \\
0 &0 &-i\epsilon^{-1}&0\\
0 &0 &0 &\frac{e^{i3\pi/4}}{\sqrt{2}\epsilon}
\end{matrix}\right),\end{array}\\
\begin{array}{@{}c@{}}
 S3\longrightarrow E3\\
 \left(\begin{matrix}
1 &0 &0 & 0 \\
0 &1 & 0 &0 \\
0 &0 &1 &0\\
0 &0 &0 &\epsilon
\end{matrix}\right).\end{array}
\end{gather*}
\textbf{S6:}
\begin{gather*}\begin{array}{@{}c@{}c@{}c@{}c@{}c@{}}
 {S}6\longrightarrow {E}12 && {S}6\longrightarrow {E}14 && {S}6 \longrightarrow {E}5\\
\left(\begin{matrix}
1 &0 &0 & 0 \\
0 &\epsilon &1 & \epsilon^{-1}\\
0 &0 & i &0\\
0 &0 &0 & i\epsilon^{-1}
\end{matrix}\right), &\qquad & \left(\begin{matrix}
1 &0 &0 & 0 \\
0 &\epsilon &0& \epsilon^{-1} \\
0 &0 & \epsilon^2 &0\\
0 &0 &0 & i\epsilon^{-1}
\end{matrix}\right),&\qquad & \left(\begin{matrix}
1 &0 &0 & 0 \\
0 &\epsilon &0& 0 \\
0 &0 &1 &0\\
0 &0 &0 &1
\end{matrix}\right),\end{array}\\
\begin{array}{@{}c@{}c@{}c@{}c@{}c@{}}
 {S}6\longrightarrow {E}4 && {S}6\longrightarrow {E}13 && {S}6\longrightarrow {E}18\\
\left(\begin{matrix}
\epsilon^2 &\epsilon &1 & \epsilon^{-1} \\
0 &\epsilon &0& 0\\
0 &0 & i &0\\
0 &0 &0 & i\epsilon^{-1}
\end{matrix}\right),  &\qquad & \left(\begin{matrix}
\epsilon &-i\epsilon &\epsilon^{-1} & 0 \\
0 &\epsilon &i \epsilon^{-1}& \epsilon^{-1} \\
0 &0 &-\epsilon^{-1} &0\\
0 &0 &0 & i\epsilon^{-1}
\end{matrix}\right), &\qquad & \left(\begin{matrix}
1 &0 &0 & 0 \\
0 &1 & 0 &0 \\
0 &0 &\epsilon^{-1} &0\\
0 &0 &0 &\epsilon
\end{matrix}\right).\end{array}
\end{gather*}
\textbf{E18:}
\begin{gather*}\begin{array}{@{}c@{}c@{}c@{}c@{}c@{}}
 {E}18\longrightarrow {E}5 && {E}18\longrightarrow {E}13 && {E}18\longrightarrow {E}4 \\
 \left(\begin{matrix}
1 &0 &0 & 0 \\
0 &\epsilon &0& 1 \\
0 &0 &1 &0\\
0 &0 &0 & 1
\end{matrix}\right),& \qquad & \left(\begin{matrix}
\epsilon &\epsilon &\frac{1}{2}\epsilon^{-1} & \frac{1}{2}\epsilon^{-1} \\
\pm i\epsilon &\mp i\epsilon & \mp \frac{i}{2}\epsilon^{-1} &\pm \frac{i}{2}\epsilon^{-1} \\
0 &0 &\frac{i}{\sqrt{2}}\epsilon^{-1} &0\\
0 &0 &0 & \frac{i}{\sqrt{2}}\epsilon^{-1}
\end{matrix}\right),&\qquad & \left(\begin{matrix}
\epsilon^2 &\epsilon &1 & (1-i)(2\epsilon)^{-1} \\
0 &\epsilon &0& (1+i)(2\epsilon)^{-1} \\
0 &0 &i &0\\
0 &0 &0 & (1+i)(2\epsilon)^{-1}
\end{matrix}\right).\end{array}
\end{gather*}
\textbf{D3E:}
\begin{gather*}\begin{array}{@{}c@{}c@{}c@{}c@{}c@{}}
  D3{E} \longrightarrow {E}5 &&  D3{E}  \longrightarrow D1D && D3{E} \longrightarrow E13\\
 \left(\begin{matrix}
1 &0 &0 & 0 \\
0 &\epsilon &1& 1\\
0 &0 &i &0\\
0 &0 &0 &\epsilon
\end{matrix}\right), &\qquad & \left(\begin{matrix}
1 &0 &0 & 0 \\
0 &-\epsilon &-\epsilon^{-1} & -1\\
0 &0 &i \epsilon^{-1} &0\\
0 &0 &0 &i
\end{matrix}\right),& \qquad & \left(\begin{matrix}
\epsilon^{-1} &\epsilon^{2} &\frac{\epsilon}{1-i} & \frac{1}{\epsilon^2(1+i)} \\
i\epsilon^{-1} &\epsilon^{2} &-\frac{\epsilon}{1-i} & \frac{1}{\epsilon^2(1-i)} \\
0 &0 &-\frac{\sqrt{2}\epsilon}{1+i}&0\\
0 &0 &0 &\frac{1}{\sqrt{2}\epsilon^2}
\end{matrix}\right),\end{array}\\
 \begin{array}{@{}c@{}c@{}c@{}c@{}c@{}}
 D3{E} \longrightarrow E4 &&  D3{E} \longrightarrow E3 &&  D3{E} \longrightarrow E18\\
 \left(\begin{matrix}
\epsilon &0 &1 & 0 \\
0 &\epsilon &0 &\epsilon^{-1} \\
0 &0 &\epsilon^2&\frac{-1+i}{\sqrt{2}\epsilon}\\
0 &0 &0 &\frac{1}{\sqrt{2}\epsilon}
\end{matrix}\right), &\qquad & \left(\begin{matrix}
1 &0 &0 & 0 \\
0 &1 & 0 &0 \\
0 &0 &1 &0\\
0 &0 &0 &\epsilon
\end{matrix}\right), &\qquad & \left(\begin{matrix}
1 &0 &0 & 0 \\
0 &1 & 0 &0 \\
0 &0 &\epsilon &0\\
0 &0 &0 &\epsilon^{-1}
\end{matrix}\right).\end{array}
\end{gather*}

\subsection[Contractions of the degenerate quadratic algebras and the lower half of the Askey scheme]{Contractions of the degenerate quadratic algebras\\ and the lower half of the Askey scheme}

\begin{figure}[h]\centering
\includegraphics[width=110mm]{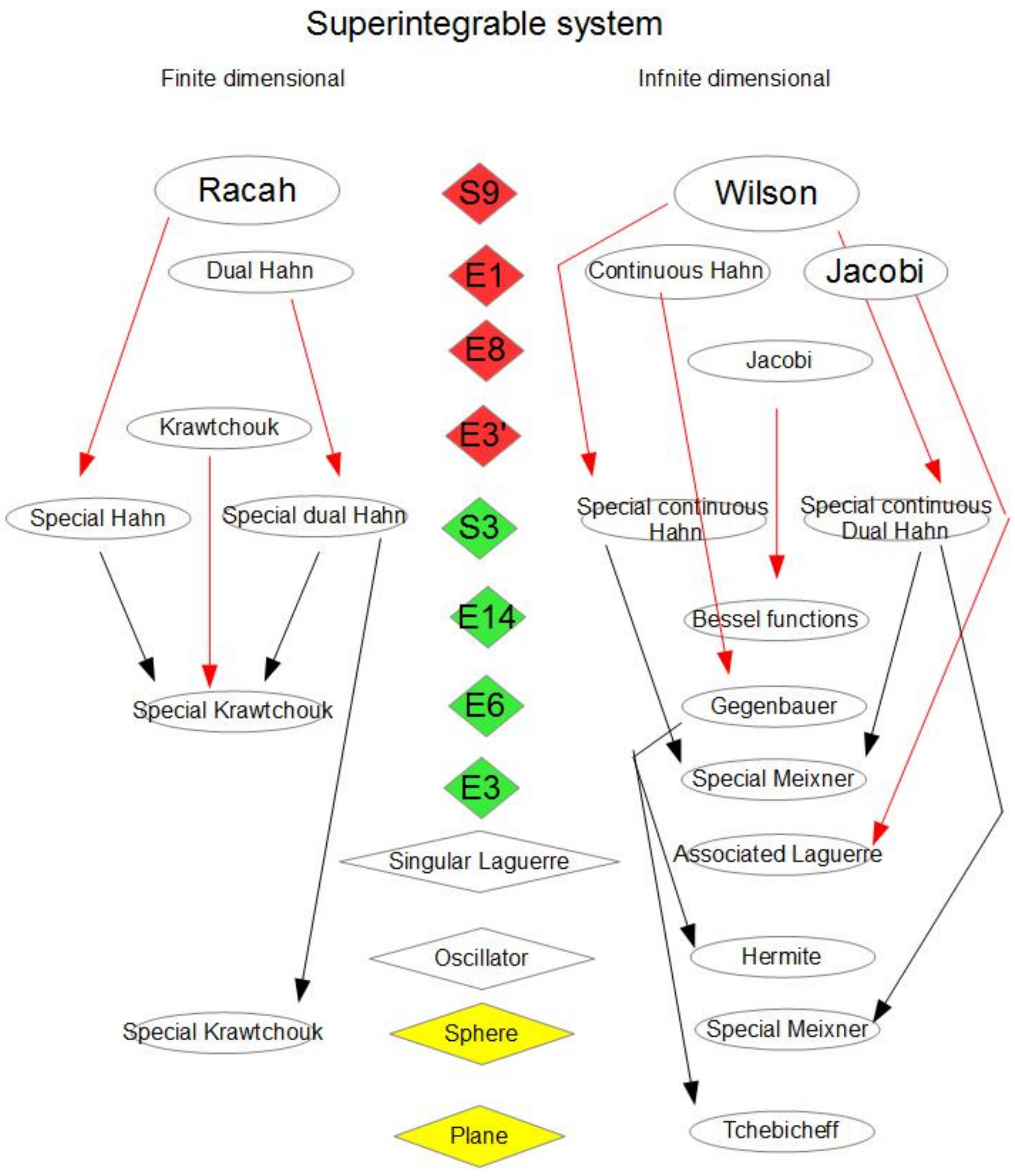}
\caption{Contractions of degenerate systems and the bottom half of the Askey scheme.}\label{caption1}
\end{figure}

The bottom half of the contraction Askey scheme relating orthogonal polynomials via contractions to degenerate, singular and free superintegrable systems is presented in Fig.~\ref{caption1}. The top half of the Scheme, relating to contractions of nondegenerate superintegrable systems can be found in~\cite{EKMS2017} and the full scheme in~\cite{KMP2014}. On the left side are the orthogonal polynomials that realize f\/inite-dimensional representations of the quadratic algebras via dif\/ference or dif\/ferential operators and on the right those that realize inf\/inite-dimensional bounded below representations. The arrows from the nondegenerate superintegrable system $S9$, $E1$, $E8$ and $E3'$ correspond to restriction/contractions to degenerate systems, i.e., the parameters in the 3-parameter potentials are restricted to the case of only 1 parameter and such that one of the symmetry operators becomes a perfect square. This increases the symmetry algebra of the resulting degenerate system.
\begin{Remark} For reference, the corresponding potentials are
\begin{itemize}\itemsep=0pt
\item $E1$: \quad $V=\alpha(x^2+y^2)+\frac{\beta}{x^2}+\frac{\gamma}{y^2}$,
\item $E3'$: \quad
$V=\alpha(x^2+y^2)+\beta x+\gamma y$,
\item $E8$:\quad
$V=\frac{\alpha (x-iy) }{(x+iy)^3}+\frac{\beta}{(x+iy)^2}+\gamma\big(x^2+y^2\big)$,
\end{itemize}
on f\/lat space and
\begin{itemize}\itemsep=0pt
\item $ S9$:\quad $V=\frac{\alpha}{s_1^2}+\frac{\beta}{s_2^2}+\frac{\gamma}{s_3^2}$,\qquad $s_1^2+s_2^2+s_3^2=1$,
\end{itemize}
on the 2-sphere.
\end{Remark}

The arrows from one degenerate superintegrable system to another are the standard contractions studied above. The singular Laguerre and oscillator superintegrable systems have singular Hamiltonians, and for these systems knowledge of the free quadratic algebra does not necessarily determine the full superintegrable system. The singular Laguerre system is a restriction/contraction of~$E1$. The resulting quadratic algebra is isomorphic to $\{ H \} \oplus
 {\mathfrak{sl}} (2,\mathbb{R})$~\cite{KMP2014}, in the same sense that the quadratic algebra of the 2D Kepler system is said to be ${\mathfrak{so}}(3,\mathbb{R})$. This is true only if the system is restricted to an eigenspace of~$H$. The Oscillator system is a contraction of $E6$ and its quadratic algebra is isomorphic to the 4-dimensional oscillator algebra~\cite{KMP2014}. The plane and sphere systems are restriction/contractions of degenerate superintegrable systems to free superintegrable systems on the plane and the 2-sphere, respectively.

\section{Conclusions and discussion}\label{section6} This paper is devoted to the study of the geometric quadratic algebras that correspond to 2D degenerate 2nd order superintegrable systems and general abstract degenerate quadratic algebras. Since the geometric quadratic algebras are uniquely determined by their free restrictions we studied only parameter-free geometric and abstract algebras. The geometric algebras were already known; in this paper we classif\/ied all abstract algebras in Table~\ref{table2}. We related the geometric and free quadratic algebras in Table~\ref{tablea}. We showed that there were 5 abstract algebras with no geometric counterpart, but that it was impossible to represent them in phase space.

In Section~\ref{Helmholtzcontractions} we derived and classif\/ied all B\^ocher contractions of 2D degenerate 2nd order superintegrable systems. In Fig.~\ref{caption1} we showed the relationship between our results and the bottom half of the Askey scheme. We derived and classif\/ied all abstract contractions of the geometric quadratic algebras, presenting the results in Table~\ref{tab2}. Comparing the B\^ocher and abstract contractions and taking into account the isomorphism of the $E14$ and $S5$ algebras, we see that there is a match except for 6 abstract contractions with no geometric realization:
\begin{gather*}
\begin{array}{@{}lll} \text{paren algebra}& \text{abstract contracted algebra}&\\
\hline
S6&D1D&\\
\hline S3&D1&\\
\hline S3&E13\\ \hline
D3E&E18&\\
\hline E12&E5&\\ \hline
E14&E5&\\
\hline
\end{array}
\end{gather*}
In two cases the failure of geometric realization is obvious: It is not possible to contract a~constant curvature space to a Darboux space~\cite{HKMS2015}. This paper is a partial warm-up for an analogous study of quadratics algebras for 3D superintegrable systems, e.g.,~\cite{CKP2015}
and for cubic algebras, e.g.,~\cite{Marquette}.

\appendix
\section{Summary of degenerate Laplace and Helmholtz systems}\label{appendixA}
The degenerate superintegrable systems can occur only on the 2-sphere, 2D f\/lat space, or one of the~4 Darboux spaces.
The notation for these systems is taken from~\cite{KKMW, KKMP}. (We write the systems in classical form; the quantum analogs have the same potentials
and the obvious replacements of classical momenta by quantum derivatives.) We assume all variables to be complex.

{\bf Degenerate complex Euclidean systems $ H=p_x^2+p_y^2+\alpha V(x,y)$:}
\begin{enumerate}\itemsep=0pt
\item $ E18$: $ V=\frac{1}{\sqrt{x^2+y^2}}$, Kepler potential,
\item $ E3$: $ V=x^2+y^2$, harmonic oscillator,
 \item $ E6$: $ V=\frac{1}{x^2}$, radial potential,
 \item $ E5$: $ V=x$, linear potential,
 \item $ E12$: $ V=\frac{ x+iy}{\sqrt{(x+iy)^2+c^2}}$,
\item $ E14$: $ V=\frac{1}{(x+iy)^2}$,
 \item $ E4$: $ V=x+iy$,
\item $ E13$: $ V=\frac{1}{\sqrt{ x+iy}}$.
\end{enumerate}
The last 4 systems are real in Minkowski space.

{\bf Degenerate systems on the complex sphere:}
We use the classical realization for $o(3,{\mathcal C})$ with basis ${ J}_1=s_2p_{s_3}-s_3 p_{s_2}$, ${ J}_2=s_3 p_{s_1}-s_1 p_{s_3}$, ${ J}_3=s_1 p_{s_2}-s_2 p_{s_1}$, and Hamiltonian ${ H}={ J}_1^2+{ J}_2^2+{ J}_3^2+\alpha V$. Here $s_1^2+s_2^2+s_3^2=1$.
\begin{enumerate}\itemsep=0pt
\item $ S6$: $V=\frac{ s_3}{\sqrt{s_1^2+s_2^2}}$, Kepler analog,
\item $ S3$: $V=\frac{ 1}{s_3^2}$, Higg's oscillator,
\item $ S5$: $ V=\frac{ 1}{(s_1+is_2)^2}$.
\end{enumerate}
The last system is real on the 2-sheet hyperboloid.

{\bf Degenerate systems on Darboux spaces:}
\begin{enumerate}\itemsep=0pt
 \item ${ D1D}$:  ${ H}=\frac{1}{4x}(p_x^2+p_y^2)+\frac{\alpha}{x}$,
\item $ D2D$: ${ H}=\frac{x^2}{x^2+1}(p_x^2+p_y^2)+\frac{\alpha}{x^2+1}$,
\item $ D3E$: $ H=\frac12\frac{e^{2x}}{e^x+1}(p_x^2+p_y^2) +\frac{\alpha}{e^x+1}$,
\item $ D4(b)D$: ${ H}=-\frac{\sin^2 2x}{2\cos 2x+b}(p_x^2+p_y^2) +\frac{\alpha}{2\cos 2x+b}$.
\end{enumerate}

\begin{Remark}\label{remark2} Every degenerate system occurs as a ``restriction'' of at least one nondegenerate system, although the symmetry algebra grows. For example the classical nondegenerate $S9$ system has the Hamiltonian
\begin{gather*} H= J_1^2+J_2^2+J_3^2+\frac{a_1}{s_1^2}+\frac{a_2}{s_2^2}+\frac{a_3}{s_3^2}\end{gather*}
and a basis of symmetries
\begin{gather*}{ L}_1={ J}_3^2+a_1\frac{s_2^2}{s_1^2}+a_2\frac{s_1^2}{s_2^2},\qquad { L}_2={ J}_1^2+a_2\frac{s_3^2}{s_2^2}+a_3\frac{s_2^2}{s_3^2},\qquad { L}_3={ J}_2^2+a_3\frac{s_1^2}{s_3^2}+a_1\frac{s_3^2}{s_1^2},\end{gather*}
where $ H={ L}_1+{ L}_2+{ L}_3+a_1+a_2+a_3$. If we let $a_1\to 0$, $a_2\to 0$ we obtain the Hamiltonian for $S3$: $H'= J_1^2+J_2^2+J_3^2+\frac{a_3}{s_3^2}$. However, now
\begin{gather*}L_1\to L_1'= J_3^2,\qquad L_2\to L_2'= J_1^2+a_3\frac{s_2^2}{s_3^2},\qquad {L}_3\to L_3' ={ J}_2^2+a_3\frac{s_1^2}{s_3^2},\end{gather*}
and the restricted system now admits the 1st order symmetry $J_3$ as well as a new 2nd order symmetry $\{J_3,L_2'\}$. These~4 symmetries are related by the Casimir. A table with all of the restrictions of nondegenerate systems on constant curvature spaces to degenerate systems can be found in~\cite{KM2014} and a table with all restrictions of nondegenerate systems on Darboux spaces can be found in~\cite{HKMS2015}.
\end{Remark}

\subsection*{Acknowledgments}
This work was partially supported by a grant from the Simons Foundation (\#~208754 to Willard Miller, Jr and by CONACYT grant (\#~250881 to M.A.~Escobar). The author M.A.~Escobar is grateful to ICN UNAM for the kind hospitality during his visit, where a part of the research was done, he was supported in part by DGAPA grant IN108815 (Mexico). We thank a~referee for pointing out the relevance of references~\cite{DufZun05, GraMarPer93,LauPiVan13}.

\pdfbookmark[1]{References}{ref}
\LastPageEnding

\end{document}